\documentclass[%
 reprint,
 amsmath,amssymb
 aps,
]{revtex4-2}

\usepackage[utf8]{inputenc}
\usepackage[T1]{fontenc}
\usepackage{amsmath,amssymb,amsthm, mathrsfs}
\usepackage{cases}
\usepackage{xcolor}
\usepackage{trimclip}
\usepackage{diagbox}
\usepackage{ragged2e}
\usepackage{caption}
\DeclareCaptionJustification{FullJustify}{\justifying}
\usepackage[normalem]{ulem}
\usepackage{booktabs}
\usepackage{multirow}
\usepackage{makecell}
\usepackage{threeparttable} 

\usepackage{enumitem}
\usepackage{ragged2e}
\usepackage{graphicx}
\usepackage{dcolumn}
\usepackage{bm}
\usepackage{hyperref}
\usepackage[colorinlistoftodos,prependcaption,textsize=tiny]{todonotes}
\usepackage{graphicx}
\usepackage{pgfplots}

\pgfplotsset{compat=newest}

\usepackage{subcaption}

\newcommand{\R}{\mathbb{R}}

\newcommand{\Bc}{\mathcal{B}}

\newcommand{\shortarrow}{\mathrel{\mathpalette\shortInternalArrow\relax}}
\newcommand{\shortInternalArrow}[2]{%
  \clipbox*{{.45\width} -2pt {\width} {\height+2pt}}{$#1\rightarrow$}%
}

\newcommand{\Oin}{\Omega_{\mathrm{in}}}
\newcommand{\Oout}{\Omega_{\mathrm{out}}}
\newcommand{\Din}{D_{\mathrm{in}}}
\newcommand{\Dout}{D_{\mathrm{out}}}
\newcommand{\cin}{c_{\mathrm{in}}}
\newcommand{\cout}{c_{\mathrm{out}}}
\newcommand{\kin}{\kappa_{\mathrm{in} \shortarrow \mathrm{out}}}
\newcommand{\kout}{\kappa_{\mathrm{out} \shortarrow \mathrm{in}}}
\newcommand{\rhoin}{\rho_{\mathrm{in}}}
\newcommand{\rhoout}{\rho_{\mathrm{out}}}
\newcommand{\dip}{\mathrm{dip\,\,depth}}

\newcommand{\Tin}{T_{\text{in}}}
\newcommand{\Tlayer}{T_{\text{layer}}}

\newcommand{\GammaInt}{\Sigma} 
\newcommand{\norm}[1]{\| #1 \|}

\newcommand{\innerH}[2]{( #1, #2 )_H}
\newcommand{\innerV}[2]{( #1, #2 )_V}

\newtheorem{theorem}{Theorem}[section]
\newtheorem{lemma}[theorem]{Lemma}

\newtheorem{definition}[theorem]{Definition}

\newcommand{\eqdef}{\stackrel{\mathrm{def}}{=}}

\usepackage{soul}
\sethlcolor{yellow}
\soulregister\textbf7
\definecolor{newgreen}{rgb}{0,0.5,0}
\newif\ifshownew
\shownewtrue

\usepackage{helvet}
\begin{document}

\preprint{APS/123-QED}

\title{Interfacial Resistance and Preferential\\ Internal Mixing of Phase-Separated Condensates}

\author{
Oihan Joyot$^{1,2,*}$,
Zoé Ferrand$^{1,*}$,
Fernando Muzzopappa$^1$,
Pierre Weiss$^{1,2}$,
Fabian Erdel$^{1,3}$\\[1ex]
$^1$\textit{\small University of Toulouse, CNRS, Center for Integrative Biology (CBI), France.}\\
$^2$\textit{\small University of Toulouse, CNRS, Institut de Recherche en Informatique de Toulouse (IRIT), France.}\\
$^3$\textit{\small Dynamics of Biological Systems, ZHMB, Saarland University, Saarbrücken, Germany.}\\
\textit{\small Correspondence: fabian.erdel@uni-saarland.de}\\[1ex]
$^*$\textit{\small Co-first authors}
}

\begin{abstract}
\noindent
Biomolecular condensates spatially organize biochemical processes. They can selectively partition molecules and enrich them well above their concentration outside. Their function depends on the exchange of molecules with their surroundings, whose kinetics are determined not only by partitioning and diffusion but also by interfacial resistance. The latter can hinder molecules from crossing the condensate boundary, thereby increasing their residence times in the condensate and promoting repeated encounters with interaction partners therein. So far, it has proven difficult to measure the interfacial resistance of condensates experimentally. Here, we combine fluorescence recovery after photobleaching (FRAP) with fluorescence correlation spectroscopy (FCS) and calibrated fluorescence imaging to disentangle partitioning, interfacial resistance, and differences in diffusion inside and outside of condensates. Applying this approach to polyelectrolyte condensates across a range of salt concentrations, we find substantial, salt-dependent interfacial resistance, which causes molecules to undergo multiple rounds of internal mixing before escaping. DDX4-YFP condensates in living cells also exhibit substantial interfacial resistance, although they are more permeable than polyelectrolyte condensates with comparable partition coefficients. Our results indicate that interfacial resistance is a tunable property of condensates that can place them in distinct transport regimes, favoring distinct functions such as dynamic reaction crucibles or insulated molecular sinks.
\end{abstract}

\maketitle

\section{\label{sec:level1}Introduction}
Biomolecular condensates are found across a wide range of biological systems, organizing molecules in time and space. Many condensates are thought to assemble through transient, attractive, multivalent interactions \cite{banani2017biomolecular, mittag2022conceptual}, forming dynamic assemblies in which molecules can diffuse. This dynamic organization raises an important question: How readily can molecules exchange between a condensate and its surroundings? Such molecular exchange is crucial for many biological processes, including biochemical reactions in which substrates enter, products leave, and molecules repeatedly encounter their potential reaction partners. The answer depends not only on how molecules diffuse inside and outside of the condensate, but also on what happens when they reach the interface. The interface can influence both where molecules preferentially localize and how quickly they cross it. Depending on its properties, molecules may exchange rapidly with the surroundings or remain confined within the condensate for extended periods of time. This distinction is functionally relevant if molecules behave differently inside and outside of the condensate, regarding their mobility, conformation, and interactions with other molecules that are differentially partitioned. In particular, interfacial resistance can slow down crossing of molecules across the interface even when they are able to diffuse rapidly within the condensate. A sufficiently large resistance will therefore increase molecular residence times in the condensate and will allow molecules to undergo multiple rounds of diffusion and reaction before escaping. Depending on this resistance, condensates may release molecules without any delay, behave as dynamic reaction crucibles retaining molecules long enough for repeated encounters that promote biochemical reactions, or act as insulated molecular sinks retaining molecules so long that they are sequestered from processes occurring outside.

Molecular diffusion within condensates and exchange with the surrounding medium have been widely studied by Fluorescence Recovery After Photobleaching (FRAP), where labeled molecules are bleached and the fluorescence recovery is followed over time \cite{taylor2019quantifying, erdel2020mouse, hubatsch2021quantitative, muzzopappa2022detecting, zhang2024the}. One variant, half-FRAP, bleaches one half of the condensate before following the recovery \cite{brangwynne2009germline, erdel2020mouse, muzzopappa2022detecting}. Such experiments provide information about how rapidly molecules redistribute within the condensate and how fast they exchange across the interface. However, the recovery dynamics contain several contributions that are difficult to distinguish experimentally. These include (i) interfacial resistance, which sets the rate at which molecules can cross the interface, (ii) partitioning, which describes the tendency of molecules to accumulate inside (or outside) of the condensate, and (iii) the mismatch between diffusion coefficients inside and outside of the condensate. Recent theoretical work has incorporated these effects into transport models for condensates \cite{hubatsch2025transport, zhang2024the, muzzopappa2022detecting, bo2021stochastic, hubatsch2021quantitative, erdel2020mouse, taylor2019quantifying}, but a central experimental challenge remains, namely to determine whether condensate interfaces possess substantial resistance and to disentangle it from partitioning and bulk diffusion. This distinction is important because partitioning alone does not determine the kinetics of molecular exchange. A molecule can strongly favor the condensate at equilibrium while still crossing its interface rapidly, or it can experience a substantial kinetic barrier at the interface and therefore remain inside of the condensate for a long time, regardless of whether it favors the condensate. Separating these two effects is essential for understanding how condensates control molecular residence times and internal mixing.

To address this issue, we extend our previously developed half-FRAP approach and combine it with Fluorescence Correlation Spectroscopy (FCS) and calibrated fluorescence imaging. We quantify the interfacial permeability $\kappa$, whose inverse corresponds to the interfacial resistance, while taking into account partitioning and diffusion. We set up a minimal model describing the recovery in both halves of a condensate, and combine it with independent measurements of partitioning and diffusion. The inferred parameters then allow us to assess how interfacial resistance contributes to molecular dynamics and in particular to internal mixing of molecules before escape. We apply this framework to coacervates formed by Poly-L-Lysine (PLL) and Hyaluronic Acid (HA) at different MgCl\textsubscript{2} concentrations, which represent a tunable \textit{in vitro} model system for biomolecular condensates \cite{park2020dehydration,yewdall2021coacervates,muzzopappa2022detecting}, and to condensates formed by the disordered model protein DDX4-YFP \cite{nott2015phase} in living cells. We find substantial interfacial resistance in both systems, which promotes preferential internal mixing. In PLL-HA coacervates, interfacial resistance decreases as salt concentration increases, consistent with a model in which molecules must overcome an energetic barrier to leave the condensate. DDX4-YFP condensates exhibit partition coefficients that are similar to those of PLL-HA coacervates at approximately 70 mM MgCl\textsubscript{2}, while their interfacial resistance is considerably lower.

Our results show that condensate interfaces can act as tunable transport barriers, allowing condensates with similar equilibrium partition coefficients to exhibit different interfacial resistances and therefore occupy distinct transport regimes that favor distinct functions.

\section{\label{sec:level2}Results}

\subsection{Minimal model for diffusive transport in a biomolecular condensate}
We first focus on a minimal model linking the interfacial permeability, the equilibrium partition coefficient and the bulk diffusion coefficients to the observables obtained by full- and half-FRAP. It provides the basis for estimating the interfacial permeability from recovery curves with independent constraints on partitioning and diffusion. We then examine how the inferred parameters determine the extent of internal mixing occurring before molecules escape from the condensate. This framework underlies both the idealized FRAP model in Section~\ref{sec:idealized-frap} and the fitting of confocal FRAP videos in Section~\ref{sec:video-fitting}.

\subsubsection{Interplay between bulk diffusion, partitioning and interfacial resistance}
\label{sec:model}
Building on previous work~\cite{scott1951diffusion,powles1992exact,grebenkov2008analytical,carr2018modelling,Moutal2019,taylor2019quantifying,schafer2020spherical,bressloff2022probabilistic,zhang2024the,hubatsch2025transport}, we represent transport through the thin interfacial layer by an effective permeability, without resolving its microscopic features. We consider a spherical domain $\Omega=\Bc(0,L)$ that can be decomposed into a smaller spherical domain $\Oin=\Bc(0,R)$ of radius $R$ that is surrounded by an annulus $\Oout = \Bc(0,L) \backslash \Oin$ of radius $L$ (Fig.~\ref{half-frap-videos}a).
We let $c(x,t)$ denote the concentration of bleached molecules at position $x$ and time $t$. We treat bleaching as a change of the fluorescent state of molecules that leaves molecular transport and the stationary condensate geometry unchanged.
We let $\cin = c|_{\Oin}$ and $\cout = c|_{\Oout}$ denote the inner and outer concentrations, respectively.
We assume that the diffusion coefficient $D$ is constant in each domain: 
\begin{equation*}
    D(x) = \begin{cases}
        D_{\textrm{in}} & \text{if } x \in \Oin, \\
        D_{\textrm{out}} & \text{if } x \in \Oout.
    \end{cases}
\end{equation*}
We consider the following diffusion equation~\cite{scott1951diffusion,taylor2019quantifying,zhang2024the}:
\begin{align}
    \frac{\partial c}{\partial t} &= \nabla \cdot (D \nabla c), \label{eq:diffusion_eq}\\ 
    -\Din \partial_r \cin|_{r=R} &= -\Dout \partial_r \cout|_{r=R}, \label{eq:flux_conservation}\\ 
    -\Din \partial_r \cin|_{r=R} &= \kappa \left( \cin(R_-) - \Gamma \cout(R_+)\right)\notag\\ &= \kin \cin(R_-) - \kout \sqrt{\Delta} \, \cout(R_+), \label{eq:robin_boundary}\\
    \Dout \partial_r \cout|_{r=L}   & = 0. \label{eq:ext_neumann}
\end{align}
The boundary condition in Eq.~\eqref{eq:flux_conservation} ensures the conservation of the flux at the interface $r=R$.
The boundary condition in Eq.~\eqref{eq:robin_boundary} models the permeation of molecules through the interface with permeability $\kappa$ (or, equivalently, with interfacial resistance $1/\kappa$) and partition coefficient $\Gamma$. The latter equals the ratio of inner and outer concentration at equilibrium. $\Delta = \Dout/\Din$ denotes the ratio of outer and inner diffusion coefficients. The coefficients multiplying $\cin$ and $\cout$ in the flux law are $\kappa$ and $\kappa\Gamma$, respectively. For comparison with the directional reflective bias defined below, we introduce the diffusion-rescaled coefficients $\kin=\kappa$ and $\kout=\kappa\Gamma/\sqrt{\Delta}$. The permeability $\kappa$ has units of length per time and determines the flux across the interface for a given concentration imbalance. Finally, the boundary condition in Eq.~\eqref{eq:ext_neumann} ensures that no flux crosses the external boundary of the outer domain at $r=L$. The model treats the interface as thin compared with the resolved length scales and assigns it an effective permeability without specifying a microscopic crossing mechanism. This sharp-interface approximation applies to spatial variations of the concentration on length scales that are large compared with the interface width.

Eq.~\eqref{eq:robin_boundary} highlights the distinction between partitioning and interfacial resistance: For any finite positive $\kappa$, zero net flux requires $\cin=\Gamma\cout$, which determines the partition coefficient $\Gamma$ in equilibrium. Away from equilibrium, $\kappa$ sets the flux for a given concentration imbalance. For $\kappa\to\infty$, the interface becomes fully permeable, $\cin(R_{-})=\Gamma\cout(R_{+})$ is enforced locally across the interface at all times, and the exchange kinetics are governed by bulk diffusion because the interface is not rate-limiting.

\begin{figure*}[tp]
\centering
\includegraphics[width=0.99\textwidth]{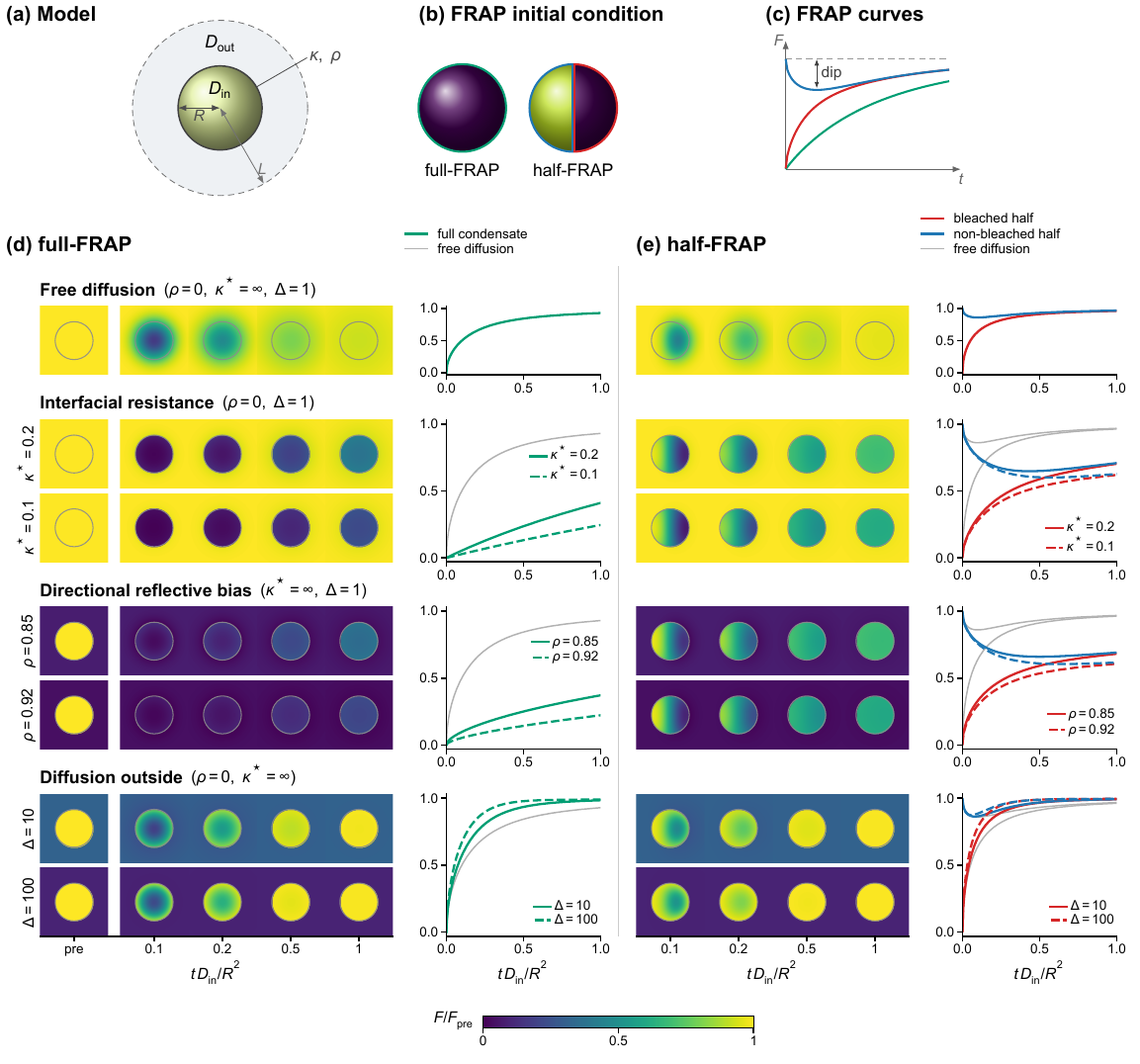}
\captionsetup{skip=6pt}
\caption{Diffusive transport model and idealized full- and half-FRAP.}
\label{half-frap-videos}
\justifying \noindent
(a) We consider diffusion of molecules in a spherical condensate of radius $R$ located within a larger spherical domain of radius $L$. The flux across the interface is set by the permeability $\kappa$ and the partition coefficient $\Gamma$ (or, equivalently, the directional reflective bias $\rho$). No flux crosses the outer boundary of the larger domain. (b) In full-FRAP, the entire condensate is bleached (green), while in half-FRAP, only one hemisphere is bleached (bleached, red; non-bleached, blue). (c) Integrated recovery curves for the entire condensate (green), the bleached half (red) and the non-bleached half (blue), whose local minimum defines the dip depth (Eq.~\eqref{eq:dip-depth}). (d,~e) Spatial distribution of bleached molecules after an instantaneous, uniform bleach of the entire condensate (d) or of one hemisphere (e) according to Eq.~\eqref{eq:integrated-curves}. Equatorial sections are shown. Gray circles denote the interface. The different rows show different parameter sets, with the corresponding recovery curves on the right (gray: free diffusion). The values of $\kappa^*$ and $\rho$ were paired to give similar recovery curves.
\end{figure*}

\subsubsection{Solution of the diffusion problem \label{sec:solution}}
 For a given initial distribution $c_0(x)=c(x,0)$, the linear transport equations determine the subsequent concentration throughout both domains. Their solution can be expressed using a weighted Green's function, or propagator,
\begin{equation}
    c(x,t)=\int_\Omega G_t(x,x')\,c_0(x')\,w(x')\,dx',
    \label{eq:general-propagation}
\end{equation}
where $w=1$ in $\Oin$ and $w=\Gamma$ in $\Oout$. This weight makes the diffusion operator self-adjoint. The spectral representation of the solution to the diffusion problem above makes the relaxation modes explicit (Appendix~\ref{spectral}). The weighted Green's function propagates the initial bleaching pattern and can be written as
\begin{equation}
    G_t(x,x') = \sum_{k} e^{-\lambda_k t}\, \Phi_k(x) \Phi^*_k(x'),
    \label{eq:green-function}
\end{equation}
where the eigenfunctions $\Phi_k$ given in Appendix~\ref{spectral} are orthonormal under the weighted scalar product, the stationary mode has the eigenvalue $\lambda_0=0$, and each nonstationary mode relaxes at rate $\lambda_k$. Spherical symmetry separates isotropic modes, which determine the integrated fluorescence signal over the entire condensate, from angular modes, which describe spatial asymmetries. The relaxation rates depend on the diffusion coefficients and the properties of the interface. The amplitudes of the modes also depend on the initial distribution $c_0$. For numerical stability, we compute the evolution by diagonalizing a symmetric discretization of the transport operator (Appendix~\ref{spectral}). For initial positions located inside of the condensate, $G_t$ coincides with the molecular transition probability density. The propagator in Eq. \eqref{eq:green-function} can be used to describe any initial distribution $c_0$ and any imaging modality, which is the basis for the following FRAP analyses.

\subsubsection{Dimensionless parameters governing diffusive transport \label{sec:reparametrization}}
To assess how the interfacial permeability influences diffusive transport in the condensate, we compare the timescale for interface crossing with that for diffusion across the condensate. Using $R$ and $R^2/\Din$ as the characteristic length and time scales, respectively, we introduce
\begin{equation}
    t^*=\frac{t\Din}{R^2},\qquad
    \Delta=\frac{\Dout}{\Din},\qquad
    \kappa^*=\frac{\kappa R}{\Din},\qquad
    a=\frac{L}{R}.
    \label{eq:dimensionless-parameters}
\end{equation}
The ratio $1/\kappa^*=\Din/(\kappa R)$ compares the interfacial timescale $R/\kappa$ with the inner diffusion timescale $R^2/\Din$ and thus quantifies interfacial resistance relative to internal transport. It can also be interpreted as an effective interfacial thickness $l = \Din/\kappa$ in units of $R$, as previously done to quantify how much the interface hinders molecules from crossing it~\cite{novikov2011random,muzzopappa2022detecting}. Together with the equilibrium partition coefficient $\Gamma$, the dimensionless parameters $\Delta$, $\kappa^*$ and $a$ govern diffusive transport and therefore determine the shape of the FRAP curves. If these dimensionless parameters are kept fixed, changing $R^2/\Din$ only rescales time. If $\Din$ or $R$ are changed while the other dimensional parameters (such as $\Dout$ or $\kappa$) are held fixed, the dimensionless parameters and consequently the shape of the FRAP curves will change.

The combined effect of partitioning and the difference in bulk diffusion coefficients is captured by the directional reflective bias
\begin{equation}
  \rho = \frac{1-\kin/\kout}{1+\kin/\kout} = \frac{\Gamma/\sqrt{\Delta}-1}{\Gamma/\sqrt{\Delta}+1},
  \label{eq:rho-definition}
\end{equation}
which is a reparametrization of $\Gamma/\sqrt{\Delta}$. It quantifies the directional asymmetry of the interface, i.e., the imbalance between the rescaled directional coefficients $\kin$ and $\kout$. Notably, a large positive $\rho$, which coincides with a stronger effective reflectivity from the inside, does not imply a small permeability $\kin=\kappa$: A large partition coefficient $\Gamma$ and a small ratio of diffusion coefficients $\Delta$ can favor retention of molecules in the condensate even when the interface is highly permeable, as molecules may return to the condensate immediately after having crossed the interface instead of escaping into the surrounding medium. Further details about the directional reflective bias and its relationship to the fluxes across the interface are provided in Appendix~\ref{reflective-bias}.

\subsubsection{Preferential internal mixing}
Following up on our previous half-FRAP studies~\cite{muzzopappa2022detecting,erdel2020mouse}, we use preferential internal mixing as a measure of molecular exploration within the condensate relative to exchange with its surroundings. Unlike prolonged residence times caused by interactions with immobile binding sites, which we neglect here, preferential internal mixing involves continued molecular diffusion within the condensate. Both slow interface crossing and returns from the surrounding phase can increase the exploration time within the condensate and thereby favor preferential internal mixing.

We consider a molecule initially located at the center of the condensate. Its mean passage time to first reach the interface is $\tau_\text{diff}=R^2/(6\Din)$. We define escape as first reaching a surface at $r=R+\delta$ (Appendix~\ref{escape-time}), and calculate the mean time $\tau_\text{res}$ accumulated inside of the condensate before the molecule has escaped, including all its returns. The time spent outside of the condensate between returns is excluded. With $\Dout^\text{eff}=\Dout(R+\delta)/\delta$ (Appendix~\ref{escape-time}), we obtain
\begin{equation}
    \tau_\text{res} = \underbrace{\frac{R^2}{6\Din}\vphantom{\frac{}{\Dout^\text{eff}}}}_{\tau_\text{diff}} + \underbrace{\frac{R}{3\kappa}\vphantom{\frac{}{\Dout^\text{eff}}}}_{\tau_\text{cross}} + \underbrace{\frac{\Gamma R^2}{3\Dout^\text{eff}}\vphantom{\frac{}{\Dout^\text{eff}}}}_{\tau_\text{outer}}.
    \label{eq:residence-time}
\end{equation}
The three terms in Eq.~\eqref{eq:residence-time} describe the initial passage to the interface, the additional internal residence time caused by finite permeability, and the additional internal residence time caused by returns from the surrounding phase occurring before a molecule reaches the surface at $r=R+\delta$ and escapes. We note that Eq.~\eqref{eq:residence-time} is similar to the characteristic relaxation time recently derived to approximate the full-FRAP recovery timescale of condensates~\cite{zhang2024the}. The dimensionless effective resistance groups the two contributions representing the excess residence time spent in the condensate:
\begin{equation}
    \frac{1}{\kappa_\text{eff}^*}
    =\underbrace{\frac{\Din}{\kappa R}}_{1/\kappa^*}
    +\underbrace{\frac{\Gamma\Din}{\Dout^\text{eff}}}_{1/\kappa_\text{outer}^*}.
    \label{eq:effective-resistance}
\end{equation}
Taking $\tau_\text{diff}$ as the characteristic duration of one diffusive exploration round, the accumulated internal residence time corresponds to $\tau_\text{res}/\tau_\text{diff}=1+2/\kappa_\text{eff}^*$ of such rounds. To define the degree of internal mixing based on the transport parameters, we use a free-diffusion reference. We calculate $\tau_\text{free}$ as the mean accumulated time spent inside the same sphere for free diffusion ($\kappa\to\infty$, $\Gamma=1$, and $\Dout=\Din$), starting at the center and using the same criterion to define escape. Subtracting $\tau_\text{diff}$ from both $\tau_\text{res}$ and $\tau_\text{free}$ isolates the excess internal residence time, yielding the degree of internal mixing
\begin{equation}
    M \eqdef \frac{(\tau_\text{res}-\tau_\text{diff})-(\tau_\text{free}-\tau_\text{diff})}{(\tau_\text{res}-\tau_\text{diff})+(\tau_\text{free}-\tau_\text{diff})} = \frac{(1+\pi)/\kappa_\text{eff}^*-1}{(1+\pi)/\kappa_\text{eff}^*+1}.
    \label{eq:internal-mixing}
\end{equation}
By construction, $M=0$ refers to the free-diffusion reference, positive $M$ indicates preferential internal mixing, and negative $M$ indicates less internal mixing compared to the free-diffusion reference. Appendix~\ref{escape-time} specifies the escape criterion and gives the explicit normalization. The analyses we conduct below estimate the transport parameters, including $\kappa$, $\Delta$ and $\Gamma$, to then determine $M$ that summarizes their combined impact on preferential internal mixing.

The distinction between interfacial and outer-domain resistance also appears in the limiting case of a quasi-steady, unbounded outer domain (Appendix~\ref{limiting-radiation}). Here, the isotropic mode obeys a radiation boundary condition with effective permeability
\begin{equation}
    \kappa_{\mathrm{eff,lim}}^{-1}
    =\kappa^{-1}+\frac{R\Gamma}{\Dout}.
    \label{eq:radiation-boundary}
\end{equation}
This is the series-resistance form of the boundary condition we considered previously~\cite{muzzopappa2022detecting}.

\subsubsection{Connecting transport parameters to molecular interactions between phase-separating molecules}
The diffusive transport model above describes diffusion across a semipermeable interface without specifying the interactions that give rise to the two phases on both sides of the interface. As many biomolecular condensates are thought to form by associative phase separation involving molecules interacting with each other, we connect the parameters used in the transport model above to the attractive interactions in such a system. We consider a sticky polymer model that has been used to describe different types of condensates \cite{choi2020physical, spruijt2013linear, rubinstein2001dynamics}, including coacervates formed by polyelectrolytes \cite{spruijt2013linear}, which we study experimentally below. Unlike continuum treatments, where interactions enter through a coarse-grained free energy, this model considers interactions at the molecular scale, connecting transport parameters to effective molecular interaction energies. This connection provides an energetic interpretation of preferential internal mixing and the half-FRAP dip studied in our previous work~\cite{erdel2020mouse,muzzopappa2022detecting}. 

Phase-separating molecules are represented as polymers containing sticky patches that can interact with each other. For electrostatic interactions, the interaction energy is given by a Yukawa potential according to
\begin{equation}
    \frac{E}{kT} = -n \frac{z^{2}l_{\text{B}}}{d} e^{-d/l_{\text{D}}}.
    \label{eq:energy_sticky_polymer}
\end{equation}
Here, $kT$ is the thermal energy, $z$ is the charge valency per sticky patch, $d$ is the effective contact separation between patches, $n$ accounts for the number of exchanged ion pairs during each rearrangement \cite{spruijt2013linear, hamad2018linear}, and $l_\text{B}$ and $l_\text{D}$ are the Bjerrum length and the Debye length that read
\begin{align}
    l_{\text{B}} &= \frac{e^{2}}{4\pi\epsilon_0\epsilon_\text{r}kT}, \label{eq:bjerrum_length}\\ 
    l_{\text{D}} &= \frac{1}{\sqrt{8\pi N_{\text{A}} l_{\text{B}}
    I(c_\text{salt})}}.
    \label{eq:debye_length}
\end{align}
Here, $e$ is the elementary charge, $\epsilon_0$ is the vacuum permittivity, $\epsilon_\text{r}$ is the relative permittivity of the solution, $N_{\text{A}}$ is the Avogadro constant, and $I(c_\text{salt})$ is the ionic strength for the salt concentration $c_{\text{salt}}$. For a divalent salt like MgCl\textsubscript{2} that we use below, $I(c_\text{salt})=3c_\text{salt}$ (with $c_\text{salt}$ given in mmol/l). For simplicity, we neglect the contribution of the Tris-Cl buffer to the ionic strength. The viscosity of the condensate can be written as \cite{spruijt2013linear,rubinstein2001dynamics}
\begin{equation}
    \eta(E) = \eta_0 e^{-E/kT}.
    \label{eq:eta_sticky_polymer}
\end{equation}
Here and below, the subscript "0" denotes quantities in the limit of negligible interactions between sticky patches ($E\to0$, corresponding to the high-salt limit).

Note that the interaction energy $E$ is negative if intermolecular interactions are attractive, so that the viscosity increases with the strength of attractive interactions. Eq. \eqref{eq:eta_sticky_polymer} implies that the viscosity is determined by the binding energy of individual sticky patches, as stress relaxation occurs via rearrangements of individual polymer segments. Diffusion coefficients are generally expected to scale with the inverse viscosity. More specifically, different diffusion mechanisms can be envisioned, which have been referred to as "walking" and "hopping" \cite{rapp2018mechanisms}. In the first case, polymers "walk" through the condensate by sequentially breaking and reforming individual interactions between individual sticky patches as they diffuse. In the second case, polymers "hop" through the condensate by breaking interactions between multiple sticky patches at once, rapidly diffusing over a certain distance, and reforming the broken interactions afterwards. To account for both possibilities, we multiply the energy term with an additional parameter $N_\eta$, which equals $N_\eta = 1$ for "walking" and $N_\eta > 1$ for "hopping". We then obtain the following expression for the diffusion coefficient inside of the condensate:
\begin{equation}
    \Din(E) = D_0 e^{N_\eta E/kT}.
    \label{eq:din_sticky_polymer}
\end{equation}
Assuming that $\Dout = D_0$ outside of the condensate, we obtain the following expression for the ratio of diffusion coefficients $\Delta$: 
\begin{equation}
    \Delta(E) = \frac{\Dout}{\Din(E)} = e^{-N_\eta E/kT}.
    \label{eq:delta_sticky_polymer}
\end{equation}
Accordingly, $\Delta$ is expected to be larger than one and to increase with the strength of attractive interactions among phase-separating molecules. The partition coefficient in the sticky polymer model can be written as
\begin{equation}
    \Gamma(E) = e^{-\Delta\mu/kT},
    \label{eq:gamma_sticky_polymer}
\end{equation}
with $\Delta\mu=\mu_\text{in}^{\circ}-\mu_\text{out}^{\circ}$ representing the difference between the standard chemical potentials inside and outside of the condensate. Assuming that each phase-separating molecule contributes the interaction energy of $N_\Gamma$ sticky patches to the chemical potential, which means that bringing the molecule from the dilute to the dense phase results in $N_\Gamma$ additional interactions between sticky patches, Eq. \eqref{eq:gamma_sticky_polymer} becomes
\begin{equation}
    \Gamma(E) = \Gamma_0 e^{-N_\Gamma E/kT}.
    \label{eq:gamma_sticky_polymer2}
\end{equation}
For simplicity, we did not explicitly consider entropic contributions from phase-separating molecules or counterions. According to Eq. \eqref{eq:gamma_sticky_polymer2}, $\Gamma$ is expected to increase with the strength of attractive interactions among phase-separating molecules. Furthermore, $\Gamma$ is larger than unity as phase-separating scaffold molecules are enriched in the condensate. Using Eq. \eqref{eq:rho-definition}, Eq. \eqref{eq:delta_sticky_polymer} and Eq. \eqref{eq:gamma_sticky_polymer2}, the directional reflective bias of the interface reads
\begin{equation}
    \rho(E) = \frac{\Gamma_0 e^{-E (N_\Gamma-\frac{1}{2}N_\eta)/kT}-1}{\Gamma_0 e^{-E (N_\Gamma-\frac{1}{2}N_\eta)/kT}+1}.
\end{equation}
For $\Gamma_0 \geq 1$ and $N_\Gamma \geq N_\eta/2$, the directional reflective bias is greater than or equal to zero, which means that the interface is effectively more reflective from the inside than from the outside. If $N_\Gamma>N_\eta/2$, then $\rho\to1$ as $E\to-\infty$, indicating an increasingly strong directional reflective bias, regardless of the permeability $\kappa$.

The relationship between the interfacial resistance and the intermolecular interaction energy is less clear, because different microscopic scenarios can be envisioned. Interfacial resistance can for example arise from an additional energy barrier that molecules have to cross at the interface, or from a mobility minimum at the interface \cite{hubatsch2025transport}. These phenomena may be influenced by long-range electrostatic effects at the interface \cite{van2024probing, majee2024charge}, which cannot be readily predicted based on the sticky polymer model.

We conclude that not all parameters used in the diffusive transport model described in Section~\ref{sec:model} are independent of each other when the model is used to describe condensates formed by associative phase separation, and that only part of the parameter space is relevant for this scenario. Both the relevant parameter space and the relationships among transport parameters depend on the microscopic model used to describe the system. Models that include, for example, entropic contributions that scale differently with interaction strength, additional repulsive interactions, or conformational changes of phase-separating molecules may yield different relationships.

\subsection{\texorpdfstring{Physical intuition from an idealized\\ FRAP model}{Physical intuition from an idealized FRAP model} \label{sec:idealized-frap}}
To illustrate how the transport parameters defined above influence the FRAP recovery, we first consider an idealized model that neglects microscope-specific effects on bleaching and imaging (Appendix~\ref{frap-predictions}). Integrating the propagator $G_t$ in Eq.~\eqref{eq:green-function} over the full condensate yields the idealized full-FRAP curve, while integrating it over each half yields the corresponding half-FRAP curves (Fig.~\ref{half-frap-videos}b--e). Full-FRAP primarily probes net exchange between the condensate and its surroundings, which is reflected by the change of the fluorescence intensity of the entire condensate. Half-FRAP creates a concentration difference between the two halves and thus probes two processes in parallel: exchange with the surroundings, reflected by changes of the intensity of the entire condensate, and internal mixing, reflected by the relaxation of the intensity difference between the two halves. Monitoring the average intensity in each half yields two complementary recovery curves that can be used jointly to infer the interfacial resistance, the partition coefficient, and the bulk diffusion coefficients. In contrast, classical full-FRAP infers these parameters from a single curve.

In the following, we use the idealized model to relate internal mixing and exchange with the surroundings to observable features of the idealized recovery curves. In the next section, we use a more realistic model that accounts for microscope-specific effects to fit spatially resolved confocal image sequences and determine the transport parameters (Section~\ref{sec:video-fitting}).

\subsubsection{Half-FRAP probes exchange and internal mixing}
We consider instantaneous, uniform bleaching of one hemisphere and a volume-integrated observation of both halves. Let $\text{FRAP}_\text{half,b}(t)$ and $\text{FRAP}_\text{half,nb}(t)$ denote the normalized fluorescence signals of the two halves, with initial values zero and one, respectively. The normalization is relative to the initially bleached population (Appendix~\ref{frap-predictions}). For the corresponding full-FRAP experiment, spherical symmetry and linearity of transport give
\begin{equation}
    \text{FRAP}_\text{full}(t) = \text{FRAP}_\text{half,b}(t) + \text{FRAP}_\text{half,nb}(t) - 1.
    \label{eq:half-full-sum}
\end{equation}
The sum of the signals in both halves reports overall recovery, equivalent to full-FRAP, whereas their difference reports the spatial imbalance. The sum depends only on isotropic modes, while the difference selects the odd angular modes excited by half-bleaching. Half-FRAP therefore probes both overall exchange and the relaxation of spatial asymmetry, which depends jointly on internal and interfacial transport. This motivates its use to characterize both processes. The identity in Eq.~\eqref{eq:half-full-sum} refers to the idealized bleaching and observation conditions. The analysis of confocal videos described in Section~\ref{sec:video-fitting} takes into account the optical properties of the microscope, so the identity is only satisfied approximately there.

\begin{figure}
\centering
\includegraphics[width=0.9\linewidth]{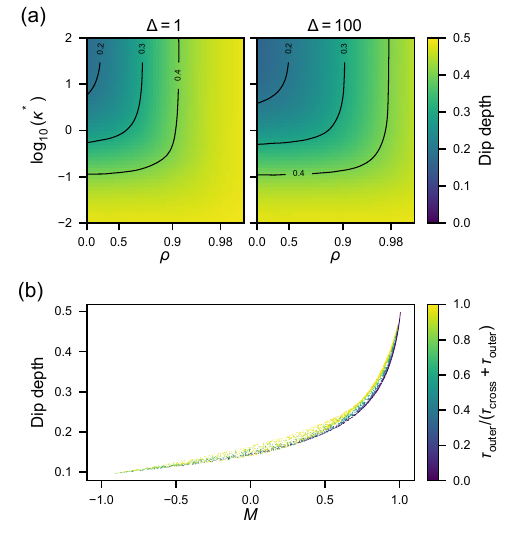}
\caption{Dip depth in the idealized FRAP model.}
\label{dip-dependence}
\justifying \noindent
(a) Dip depth for different combinations of directional reflective bias $\rho$, dimensionless interfacial resistance $1/\kappa^*$, and ratio of diffusion coefficients $\Delta$. The $\rho$ axes use an arctanh transform. (b) The approximate collapse against the degree of internal mixing $M$ relates the dip depth, a visible feature of half-FRAP, to internal mixing. Points correspond to 6,000 parameter sets drawn log-uniformly with $\kappa^*\in[10^{-3},10^{3}]$ and $\Gamma,\Delta\in[1,100]$, colored by the fraction of the excess residence time due to returns from outside, $\tau_\text{outer}/(\tau_\text{cross}+\tau_\text{outer})$.
\end{figure}

\subsubsection{The dip depth reports on preferential internal mixing \label{sec:dip-depth}}
Internal mixing can transiently lower the fluorescence of the initially non-bleached half as bleached molecules enter it from the bleached half. This produces an intensity dip (Fig.~\ref{half-frap-videos}c) whose depth is defined by~\cite{muzzopappa2022detecting}
\begin{equation}
    \dip = 1-\min_{t\geq0}\text{FRAP}_\text{half,nb}(t).
    \label{eq:dip-depth}
\end{equation}
The dip reflects a balance between redistribution of molecules within the condensate interior and exchange across its interface. At the intensity minimum, the corresponding net flux contributions cancel (Appendix~\ref{integrals}). Slow interface crossing or returns from the outside allow bleached molecules to diffuse for extended times within the condensate, favoring their accumulation in the non-bleached half and deepening the dip (Fig.~\ref{half-frap-videos}).

Across the parameter combinations examined, the dip increases with directional reflective bias $\rho$ and interfacial resistance $1/\kappa^*$, while the ratio of diffusion coefficients has a milder effect for $\Delta\geq1$, which is the parameter range we focused on as biomolecular condensates are typically more viscous than the surrounding phase (Fig.~\ref{dip-dependence}a; recovery curves in Fig.~\ref{half-frap-parameter-influence}). When plotting the dip depth against the degree of internal mixing $M$, we observed an approximate collapse to a monotonic master curve (Fig.~\ref{dip-dependence}b). Thus, the degree of internal mixing is connected to an observable feature of half-FRAP, with a deeper dip corresponding to stronger preferential internal mixing. At fixed $\Gamma$, $\Delta$, $\kappa^*$, and $L/R$, changing the time scale $R^2/\Din$ changes the timescale of the recovery without changing the dip depth.

The dip--$M$ relationship and Eq.~\eqref{eq:internal-mixing} provide an approximate estimate of the dimensionless effective resistance $1/\kappa_\text{eff}^*$, which includes contributions from the interfacial resistance and the resistance of the outer domain that determines the number of returns from outside (Eq.~\eqref{eq:effective-resistance}). Because the dip--$M$ relation shows some residual spread (Fig.~\ref{dip-dependence}b), the dip cannot be exactly converted into $1/\kappa_\text{eff}^*$. Accurately separating the contributions of the interfacial resistance and the resistance of the outer domain requires additional information, namely the partition coefficient and the ratio of diffusion coefficients. In an experimental FRAP measurement, the dip depth will also depend on the initial distribution $c_0$ and on the imaging modality. We therefore use the spatial and temporal information contained in the experimental FRAP data to estimate the transport parameters, as explained below (Section~\ref{sec:video-fitting}).

\subsection{Interfacial transport in PLL-HA and DDX4-YFP condensates \label{sec:video-fitting}}

We applied the theoretical framework developed above to coacervates formed by the polyelectrolytes PLL and HA, a tunable model system for phase-separated biomolecular condensates \cite{muzzopappa2022detecting, park2020dehydration,yewdall2021coacervates}. In particular, the intermolecular interaction strength in the PLL-HA system can be tuned by addition of salt, which screens the electrostatic interactions between positively charged PLL and negatively charged HA molecules. The salt dependence of the interaction energy in Eq. \eqref{eq:energy_sticky_polymer} enters through the Debye length, which is proportional to $1/\sqrt{c_\text{salt}}$. As described below, we experimentally measured the salt-dependent partition coefficient $\Gamma$ of PLL-HA coacervates via calibrated fluorescence microscopy imaging and the ratio of diffusion coefficients $\Delta$ via fluorescence correlation spectroscopy (FCS), which allowed us to determine the directional reflective bias of their interface. We then used quantitative half-FRAP in conjunction with the known values of $\Gamma$ and $\Delta$ to determine interfacial resistances, yielding a comprehensive description of diffusive molecular transport across the coacervate interface. Finally, we used our approach to determine the interfacial resistance of condensates formed by the disordered nuage protein DDX4-YFP \cite{nott2015phase} in living mammalian cells.

\subsubsection{Partitioning and bulk diffusion in PLL-HA coacervates are salt-dependent}

\begin{figure}
\centering
\includegraphics[width=1\linewidth]{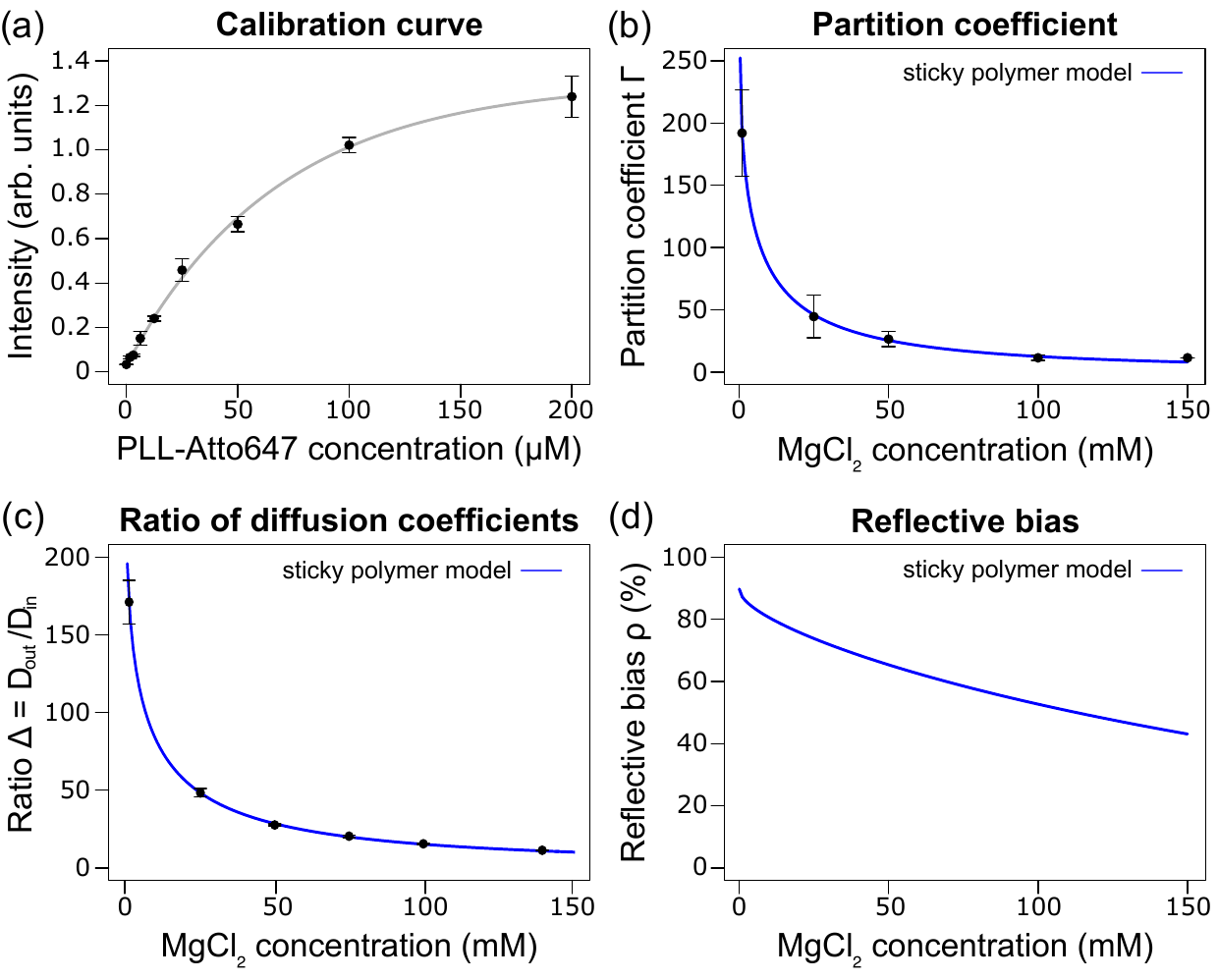}
\caption{Partition coefficients and ratios of diffusion coefficients for PLL-HA coacervates.}
\label{pll_ha_gamma_eta}
\justifying \noindent
(a) Calibration curve relating fluorescence intensities in microscopy images to Atto647N-labeled PLL concentrations. (b) Partition coefficients $\Gamma$ obtained from the analysis of microscopy images of PLL-HA coacervates (black). The blue line represents the fit with Eq. \eqref{eq:gamma_sticky_polymer2}. (c) Ratio of diffusion coefficients obtained from FCS measurements conducted in the center of the coacervates and outside of the coacervates. The blue line represents the fit with Eq. \eqref{eq:delta_sticky_polymer}. (d) Directional reflective bias $\rho$ obtained from the partition coefficients and ratios of diffusion coefficients in the previous panels.
\end{figure}

We first sought to determine the equilibrium partition coefficient $\Gamma$ for PLL-HA coacervates at different salt concentrations based on confocal microscopy images (Section~\ref{methods:imaging}). We imaged Atto647N-labeled PLL solutions at different concentrations to establish the relationship between fluorescence intensity and concentration. As expected, the relationship is approximately linear for low concentrations and becomes non-linear for higher concentrations, probably due to quenching effects (Fig. \ref{pll_ha_gamma_eta}a). We subsequently used this calibration curve to convert the intensities inside and outside of the coacervates into concentrations, yielding the partition coefficients shown in Fig. \ref{pll_ha_gamma_eta}b. As expected from Eq. \eqref{eq:gamma_sticky_polymer2}, the partition coefficient decreases with increasing salt concentrations as electrostatic interactions are progressively screened.

\begin{figure}[b]
\centering
\includegraphics[width=\linewidth]{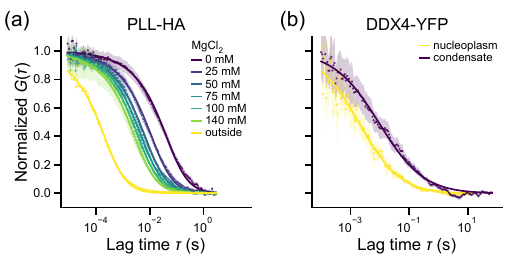}
\caption{FCS experiments for PLL-HA and DDX4-YFP.}
\label{pll_ha_fcs}
\justifying \noindent
Autocorrelation curves for PLL-HA coacervates at the indicated conditions (a) and for DDX4-YFP condensates in living cells (b). Curves were normalized with respect to the apparent particle number $N$ obtained from fitting. Lines in panel a are fits to the normal diffusion model in Eq. \eqref{eq:fcs}, lines in panel b are fits to the anomalous diffusion model in Eq. \eqref{eq:fcs2}.
\end{figure}

\begin{figure*}[tp]
\centering
\includegraphics[width=\textwidth]{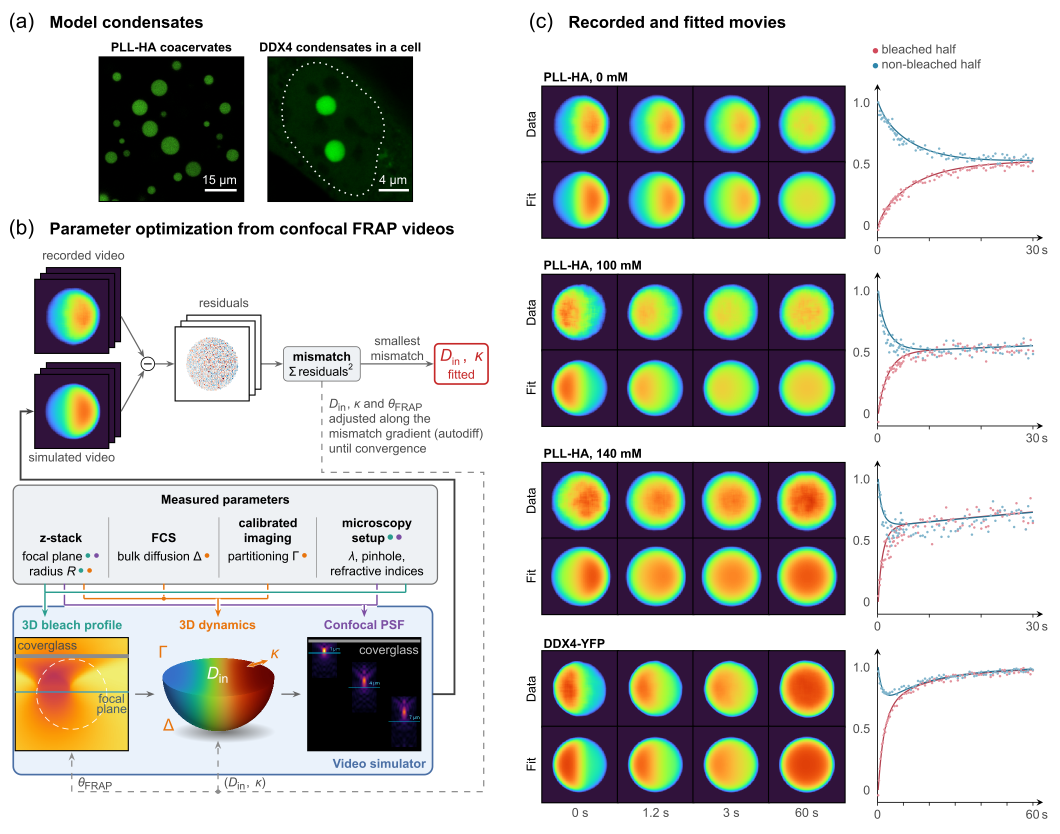}
\caption{Inner diffusion coefficient and interfacial resistance inferred from confocal FRAP videos.}
\justifying \noindent
(a)~Model condensates: PLL-HA coacervates reconstituted \textit{in vitro}, and DDX4-YFP condensates in a living cell (dotted line: outline of the cell nucleus).
(b)~Fitting procedure. A video simulator predicts the confocal FRAP movie based on three components: the three-dimensional bleach profile, set by the bleaching parameters $\theta_\mathrm{FRAP}$; diffusion inside of the condensate and exchange across its interface, set by $\Din$ and $\kappa$; and confocal image formation with the depth-dependent point-spread function (PSF).
Its other inputs are measured independently (colored dots illustrate where each one enters). The simulated movie is compared to the recorded one, and $\Din$, $\kappa$ and $\theta_\mathrm{FRAP}$ are adjusted by gradient descent (Section~\ref{sec:video-fitting}).
(c)~Recorded (Data) and simulated (Fit) frames for PLL-HA coacervates at the indicated salt concentrations and for a DDX4-YFP condensate. Dots: data; lines: recovery curves computed from the simulated videos. For display only, recorded frames were denoised, and neighboring condensates were masked; raw images were used for fitting. The fit results are listed in Table~\ref{fit-parameters}.
\label{fig:video-overview}
\end{figure*}

We next sought to determine the ratio of diffusion coefficients outside and inside of PLL-HA coacervates at different salt concentrations. To this end, we conducted FCS measurements at the center of coacervates and outside of coacervates (Fig. \ref{pll_ha_fcs}a; Section~\ref{methods:fcs}). Both measurements were conducted far away from the coacervate interface. The resulting autocorrelation curves could be adequately fitted with the diffusion model in Eq. \eqref{eq:fcs}. As expected from Eq. \eqref{eq:delta_sticky_polymer}, the ratio of diffusion coefficients $\Delta$ decreases with increasing salt concentrations (Fig. \ref{pll_ha_gamma_eta}c). We jointly fitted the measured partition coefficients and ratios of diffusion coefficients to the predictions from the sticky polymer model described above (blue lines in Fig. \ref{pll_ha_gamma_eta}b,c), yielding good fits with an apparent energy $N_\Gamma E \approx -5.5 \, kT$ (in the absence of MgCl\textsubscript{2}) and $N_\eta \approx N_\Gamma$.

Having determined the partition coefficients and ratios of diffusion coefficients for PLL-HA coacervates, we next derived the directional reflective bias of the interface $\rho$ (Fig. \ref{pll_ha_gamma_eta}d). In the absence of MgCl\textsubscript{2}, it amounted to $\rho \approx 90\%$, which means that the interface strongly favors retention of molecules inside the condensate. With increasing salt concentrations, $\rho$ decreased and reached roughly 50\% at 140 mM MgCl\textsubscript{2}, which is close to the critical concentration where coacervates dissolve. The normalized directional weights $\rhoin$ and $\rhoout$ defined in Appendix~\ref{reflective-bias}, which further characterize the asymmetry of the interface, amount to approximately $94\%$ and $6\%$ without MgCl\textsubscript{2}, and $77\%$ and $23\%$ at 140 mM MgCl\textsubscript{2} (Table~\ref{fit-parameters}). These results show that the interfaces of PLL-HA coacervates exhibit strong directional asymmetry linked to partitioning and the mismatch of diffusion coefficients inside and outside of the coacervates.

\subsubsection{Diffusion coefficients and interfacial resistances of PLL-HA coacervates are salt-dependent}
We next developed a numerical method to estimate the inner diffusion coefficient and the interfacial resistance of PLL-HA coacervates from confocal FRAP videos. The method fits entire confocal image sequences. The forward model evolves the three-dimensional fluorescence distribution using the transport equations of Section~\ref{sec:model}, and maps it to the recorded focal plane, accounting for excitation and detection. The axial position of the imaging plane is determined before fitting from the prebleach image and a $z$-stack.
The initial concentration $c_0$ after the bleach is estimated in the fit and described using a three-parameter model, which accounts for the dose, a residual fraction of fluorescence, and an effective bleaching order (Section~\ref{methods:video-fitting}). 

We jointly estimated the diffusion coefficient $\Din$ and the permeability $\kappa$ by minimizing the discrepancy between predicted and recorded images over the observed pixels and recovery times. The discrepancy's gradient was computed by automatic differentiation. The parameters $\Din$ and $\kappa$ are shared by all half-bleached condensates reconstituted at a given MgCl\textsubscript{2} concentration. Numerical and experimental details are given in Methods (Sections~\ref{methods:frap} and~\ref{methods:video-fitting}). Half-averaged curves and dip depths summarize the fitted dynamics and connect the experimental measurements to the idealized recovery dynamics discussed in Section~\ref{sec:idealized-frap}. Figure~\ref{fig:video-overview} illustrates the inference workflow and shows examples of recorded and reconstructed movies, while Table~\ref{fit-parameters} reports the fitted transport parameters.

The fit results for the interfacial resistance $1/\kappa$ are plotted in Fig. \ref{pll_ha_half_frap}a. The interfacial resistance decreases with increasing salt concentrations, consistent with the interpretation that it arises from electrostatic interactions established by the molecules in the coacervates. We fitted the interfacial resistance with the expression

\begin{figure}
\centering
\includegraphics[width=1\linewidth]{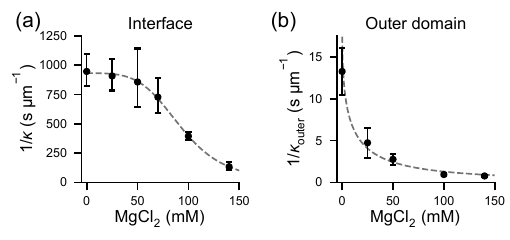}
\caption{Interfacial resistance of PLL-HA coacervates}
\label{pll_ha_half_frap}
\justifying \noindent
(a) Interfacial resistance $1/\kappa$. (b) Outer-domain resistance $1/\kappa_{\text{outer}}=\bar R \Gamma / D_{\text{out}}^{\text{eff}}$, where $\bar R$ is the mean radius of the half-bleached condensates. Error bars: 95\% jackknife confidence intervals in (a); propagated uncertainties of $\Gamma$, $\Delta$ and $\Din$ in (b). Dashed lines: in (a), fit of Eq.~\eqref{eq:kappa-fit} with the interaction energy fixed to the value obtained from the fits in Fig.~\ref{pll_ha_gamma_eta}; in (b), combination of the fitted exponential laws of $\Gamma$, $\Delta$ and $\Din$, with no free parameter.

\end{figure}

\begin{equation}
    \kappa^{-1}(E) = \left[\kappa_0 + \kappa_1 e^{N_\kappa E/kT}\right]^{-1},
    \label{eq:kappa-fit}
\end{equation}
where the energy scale $N_\Gamma E\approx-5.5\,kT$ determined above was used to fix $E$ at each salt concentration. Eq. \eqref{eq:kappa-fit} describes the scenario that molecules have to overcome an energy barrier of height $N_\kappa E$ when leaving the condensate, while a parallel route with low permeability $\kappa_0$ exists that does not depend on the interaction energy. The fit yields $N_\kappa \approx 4.5\, N_\Gamma$, $\kappa_0 \approx 1 \mathrm{\,nm/s}$ and $\kappa_1 \approx 300\mathrm{\,\mu m/s}$ (dashed line in Fig.~\ref{pll_ha_half_frap}a). To put the value of the interfacial resistance $1/\kappa$ into context, we calculated the outer-domain resistance $1/\kappa_\text{outer}$. It compares diffusive mixing in the inner domain to diffusive escape in the outer domain. As shown in Fig. \ref{pll_ha_half_frap}b, $1/\kappa_\text{outer}$ is much smaller than the interfacial resistance, which indicates that particle escape from PLL-HA coacervates is mainly limited by the interfacial resistance rather than the resistance of the outer domain. We also calculated the degree of internal mixing $M$, yielding $M\geq99\%$ for all conditions (Table~\ref{fit-parameters}). Based on Eqs.~\eqref{eq:residence-time}-\eqref{eq:internal-mixing}, this means that molecules undergo more than a hundred exploration rounds $\tau_{\text{res}}/\tau_{\text{diff}}$ before they escape from the coacervate.

Accordingly, we find a pronounced salt-dependence of the interfacial resistance, with coacervates at low salt concentrations that are deep in the two-phase regime exhibiting a strong resistance and coacervates close to the coexistence line between the one- and the two-phase regime exhibiting a much lower resistance. This difference can affect any process within the condensates that depends on molecular exchange, by altering the rates at which molecules enter and leave.

\begin{figure}[b]
\centering
\includegraphics[width=\linewidth]{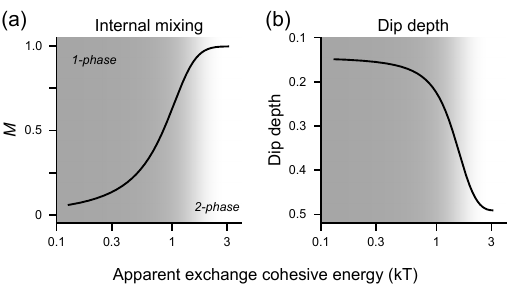}
\caption{Link between preferential internal mixing and energetics of PLL-HA coacervates}
\label{dip-energy}
\justifying \noindent
Expected degree of internal mixing (a) and dip depth (b) based on the relationships between $\Gamma$, $\Delta$, $\kappa^*$ and the apparent exchange cohesive energy. Gray areas denote the 1-phase regime where no PLL-HA coacervates are formed. Dip depths are calculated according to the idealized FRAP model. The extrapolation to larger cohesive energies depends on the exact value of $\kappa_0$ in Eq. \eqref{eq:kappa-fit} and on the onset of solid-like behavior, beyond which coacervate properties may deviate from the predictions above.
\end{figure}

\justifying
\subsubsection{Link between interaction energy and preferential internal mixing for PLL-HA coacervates}
Having determined the salt-dependence of the partition coefficient $\Gamma$, the ratio of diffusion coefficients $\Delta$ and the interfacial resistance $1/\kappa$ for PLL-HA coacervates, we revisited the link between preferential internal mixing and the energetics of the PLL-HA system. Using the expressions for $\Gamma(E)$ in Eq. \eqref{eq:gamma_sticky_polymer2}, $\Delta(E)$ in Eq. \eqref{eq:delta_sticky_polymer} and $\kappa^{-1}(E)$ in Eq. \eqref{eq:kappa-fit}, we plotted the expected degree of internal mixing $M$ and the expected dip depth (according to the idealized model) as a function of the apparent exchange cohesive energy $-N_\Gamma E$, with $E<0$ for attractive interactions (Fig. \ref{dip-energy}). We obtained a sigmoidal relationship that resembles the one we previously observed when plotting the dip depth against the apparent interfacial energy \cite{muzzopappa2022detecting}. The inflection point occurs at an apparent cohesive energy of $\sim 1.6\, kT$, while it occurred at a much smaller apparent interfacial energy of $\sim 0.03\, kT$ \cite{muzzopappa2022detecting}, which suggests that changing the interfacial area of a coacervate costs much less energy than exchanging a molecule between the coacervate and the dilute phase. This is plausible if far fewer cohesive contacts need to be broken in the first case than in the second, which is in line with the ultralow interfacial tensions of coacervates despite their comparatively high viscosities and partition coefficients.

\subsubsection{DDX4-YFP condensates in living cells have lower interfacial resistances than PLL-HA coacervates}
We next applied our approach to DDX4-YFP condensates in living NIH/3T3 cells (Fig.~\ref{fig:video-overview}a,c). DDX4 is a well-characterized disordered model protein that phase-separates on its own and forms condensates in cells. It is a core component of nuage granules, and its N-terminal disordered region has been shown to be sufficient to drive phase separation \cite{nott2015phase}. We determined the ratio of diffusion coefficients for DDX4-YFP condensates in living cells by FCS (Fig.~\ref{pll_ha_fcs}b), and we subsequently fitted individual confocal half-FRAP videos of DDX4-YFP condensates in living cells using the measured partition coefficients, the fluorescence of the surrounding nucleoplasm, and an additional immobile fraction (Sections~\ref{methods:ddx4} and~\ref{methods:video-fitting}). We obtained a median inner diffusion coefficient $\Din=0.4\,\mu\mathrm{m}^2\mathrm{s}^{-1}$, a median permeability $\kappa=0.15\,\mu\mathrm{m\,s}^{-1}$, and a median degree of internal mixing $M=86\%$. This corresponds to about seven rounds of internal exploration before escape, compared with more than a hundred for PLL-HA coacervates. These results show that DDX4-YFP condensates exhibit preferential internal mixing, retaining molecules for extended periods of time before they can escape. At the same time, DDX4-YFP condensates are considerably more permeable than PLL-HA coacervates at every salt concentration tested. These results point to a possible trend toward more permeable condensate interfaces in cells, consistent with earlier work suggesting that cellular condensates differ from their \textit{in vitro} counterparts with respect to their exchange kinetics \cite{taylor2019quantifying,muzzopappa2022detecting}.

\begin{table}[b]
\setlength{\tabcolsep}{1.5pt}
\renewcommand{\arraystretch}{1.15}
\setlength{\aboverulesep}{0pt}\setlength{\belowrulesep}{0pt}\setlength{\extrarowheight}{1pt}
\newcolumntype{I}{!{\hspace{1.3pt}\vrule\hspace{1.3pt}}}
\newcommand{\ci}[2]{\makecell{#1\\[-2pt]{\tiny #2}}}
\newcommand{\hd}[2]{\makecell{#1\\[-1pt]{\scriptsize #2}}}
\newcommand{\pmt}[2]{\makecell{#1\\[-2pt]{\tiny $\pm$#2}}}
\footnotesize
\begin{tabular}{@{}l ccIccIcccIc@{}}
\toprule
 & \hd{$\Din$}{($\mu$m$^2$\,s$^{-1}$)} & \hd{$\kappa$}{(nm\,s$^{-1}$)}
 & $\Gamma$ & $\Delta$ & \hd{$\rho$}{(\%)} & \hd{$\rhoin$}{(\%)} & \hd{$\rhoout$}{(\%)} & \hd{$M$}{(\%)} \\
\midrule
\multicolumn{9}{@{}l}{\scriptsize\textit{PLL-HA: joint fit per condition}}\\
0 mM   & \ci{0.2}{0.18--0.25} & \ci{1.1}{0.91--1.22} & \pmt{192}{35} & \pmt{171}{14} & \pmt{87}{2} & \pmt{94}{1} & \pmt{6}{1} & $>99$ \\
25 mM  & \ci{0.4}{0.38--0.48} & \ci{1.1}{0.95--1.27} & \pmt{45}{17}  & \pmt{48}{3}   & \pmt{73}{9} & \pmt{87}{4} & \pmt{13}{4} & $>99$ \\
50 mM  & \ci{0.7}{0.61--0.72} & \ci{1.2}{0.87--1.56} & \pmt{26}{6}   & \pmt{28}{1}   & \pmt{66}{7} & \pmt{83}{3} & \pmt{17}{3} & $>99$ \\
70 mM  & \ci{1.08}{1.04--1.13} & \ci{1.4}{1.12--1.68} & $18^{*}$      & $22^{*}$    & $58^{*}$    & $79^{*}$    & $21^{*}$    & $>99$ \\
100 mM & \ci{1.3}{1.18--1.40} & \ci{2.5}{2.32--2.77} & \pmt{11}{2}   & \pmt{16}{3}   & \pmt{47}{8} & \pmt{73}{4} & \pmt{27}{4} & $>99$ \\
140 mM & \ci{2.1}{1.84--2.36} & \ci{7.5}{5.86--9.50} & \pmt{11}{1}   & \pmt{11}{1}   & \pmt{54}{4} & \pmt{77}{2} & \pmt{23}{2} & $>99$ \\
\addlinespace[2pt]
\multicolumn{9}{@{}l}{\scriptsize\textit{DDX4-YFP: individual fits ($n=19$), median}}\\
       & \ci{0.4}{0.34--0.45} & \ci{145}{122--174}    & \ci{16}{14.6--16.0} & \pmt{6}{1} & \ci{73}{71--73} & \ci{86}{86--87} & \ci{14}{13--14} & \ci{86}{81--89} \\
\bottomrule
\end{tabular}
\caption{   
Transport parameters for PLL-HA and DDX4.\vspace{0.05cm}\\
\small Diffusion coefficient $\Din$ and interfacial permeability $\kappa$ from confocal FRAP videos, partition coefficient $\Gamma$ from calibrated imaging, and ratio of diffusion coefficients $\Delta$ from FCS. The directional reflective bias $\rho$, the normalized directional weights $\rhoin$ and $\rhoout$, and the degree of internal mixing $M$ were then derived from the parameters above.
Each value is reported with its uncertainty below: 95\% confidence intervals for $\Din$ and $\kappa$, and $\pm$ standard error for $\Gamma$, $\Delta$ and the quantities derived from them. $M$ is the lower bound of its confidence interval,
evaluated at the mean condensate radius of each condition. $^{*}$, interpolated.}
\label{fit-parameters}
\end{table}

\justifying
\section{\label{sec:level3}Discussion}
In this work, we investigate how the transport properties of condensate interfaces regulate internal mixing and molecular exchange between condensates and their surroundings. By combining quantitative half-FRAP with calibrated fluorescence imaging and fluorescence correlation spectroscopy, we establish an experimental approach to disentangle interfacial permeability, equilibrium partitioning, and differences in bulk diffusion. We apply it to PLL-HA coacervates, a tunable \textit{in vitro} model system for associative phase separation \cite{park2020dehydration,yewdall2021coacervates}, and to DDX4-YFP condensates \cite{nott2015phase} in living cells. Both systems exhibit preferential internal mixing, but their interfaces differ substantially in permeability. In particular, PLL-HA coacervates exhibit pronounced, salt-dependent interfacial resistance, whereas the interfacial resistance of DDX4-YFP condensates is much lower. These findings show that condensates with similar molecular enrichment can have markedly different interfacial transport properties.

In PLL-HA coacervates, interfacial resistance decreases with increasing MgCl\textsubscript{2} concentration, from $\sim$1 s/nm at low salt to $\sim$0.1 s/nm at high salt, near the critical point. Over the same range, the directional reflective bias also decreases, from $\sim$90\% to $\sim$50\%. Thus, increasing MgCl\textsubscript{2} affects both the kinetic resistance to interface crossing and the directional asymmetry arising from partitioning and bulk diffusion. Despite these changes, a high degree of internal mixing is maintained across the salt concentrations we investigated here. Under all conditions tested, the additional residence time arising from interfacial resistance exceeds both the characteristic time for diffusion within the condensate and the additional residence time associated with transient excursions in the surrounding medium. Interfacial resistance therefore constitutes the dominant kinetic bottleneck for molecular exchange with the surroundings.

The side-by-side comparison between PLL-HA and DDX4-YFP condensates highlights the distinction between equilibrium partitioning and interfacial resistance. At approximately 70 mM MgCl\textsubscript{2}, PLL-HA coacervates exhibit a partition coefficient that is similar to that of DDX4-YFP condensates, yet their interfacial resistance is about two orders of magnitude larger. DDX4-YFP condensates therefore exchange molecules across their interfaces much faster, despite exhibiting a similar degree of molecular enrichment. They nevertheless also show preferential internal mixing. Accordingly, similar partition coefficients do not necessarily imply similar exchange kinetics. The differences between both classes of condensates is in line with previous suggestions that cellular and reconstituted condensates tend to differ in their interfacial properties \cite{muzzopappa2022detecting,taylor2019quantifying}. Further experiments across a wider range of condensate systems will be needed to establish how general this trend is.

For PLL-HA coacervates, the salt dependence of the interfacial resistance is consistent with intermolecular interactions regulating molecular exchange, with stronger interactions increasing the resistance. The apparent exchange cohesive energies inferred here are substantially larger than the apparent interfacial energies determined previously \cite{muzzopappa2022detecting}. This difference suggests that molecular exchange may require the breakage of more cohesive contacts than are involved in changing the interfacial area, consistent with the ultralow interfacial tensions of coacervates despite their comparatively high viscosities and partition coefficients. However, our analysis provides only an effective energetic interpretation of the salt-dependent interfacial resistance, rather than identifying its microscopic origin. The resistance could arise from an additional energy barrier that molecules must overcome, but also from reduced mobility within the interfacial region or other effects associated with the organization of molecules near the interface \cite{hubatsch2025transport,majee2024charge}. Our model represents condensates as spherical compartments separated from their surroundings by a sharp interface with an effective permeability, and therefore does not resolve the microscopic structure or finite width of the interface. Real condensates may exhibit structural heterogeneity at the interface or may contain an internal mesh-like network, none of which are explicitly included in the present model. Future work will establish how these features affect the effective interfacial resistance assessed here.

The interfacial transport properties of condensates assessed here are likely linked to condensate function. By increasing molecular residence times and suppressing transient excursions into the surrounding medium, large interfacial resistance can maintain a quasi-isolated pool of molecules within a condensate. This could affect target-search processes by promoting repeated encounters between molecules in the condensate. Prolonged residence times may also facilitate localized biochemical memory that allows molecules to become functionally distinct from those in the surrounding environment, for example through conformational changes or chemical modifications that occur and revert on timescales comparable to the internal residence time. In addition, changes in interfacial resistance could also regulate the exchange of substrates, products and regulatory factors, modulating biochemical reactions in active condensates. The principles established here may extend to a broader range of cellular and synthetic \textit{in vitro} condensates, in which chemical modulation of intermolecular interactions could be used to tune interfacial resistance to enable, for example, controlled drug release.

\begin{acknowledgments}
	We thank Frédéric de Gournay for support and helpful discussions. OJ is supported by a PhD fellowship from the MITT doctoral school and the Occitanie Region as part of the AI for Health Program. Part of this work was funded by the ERC-2024-CoG RENOME and the ANR CLEAR Microscopy (ANR-25-CE45-3780). This work was performed using HPC resources from GENCI-IDRIS (Grant AD011012210). 
	
\end{acknowledgments}
Data and code are available at : \url{https://gitlab.irit.fr/mambo/oihan_joyot/frap_spectral_simulation}

\begin{figure}
\centering
\includegraphics[width=\linewidth]{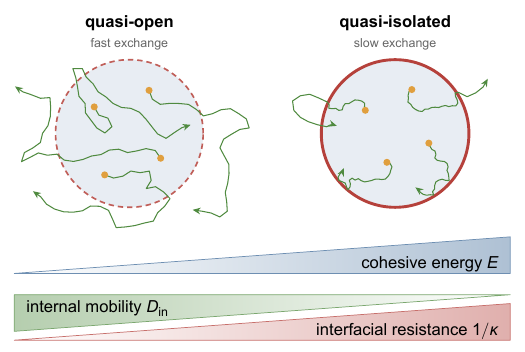}
\vspace{0.5pt}
\caption{Distinct dynamic signatures of condensates}
\label{conclusion}
\justifying \noindent
Biomolecular condensates can occupy distinct transport regimes. Condensates with larger interfacial resistances contain a quasi-isolated pool of molecules, as exchange across their interface is slow. This can affect target search processes and localized biochemical memory. In the PLL-HA system, larger interfacial resistances are linked to stronger intermolecular interactions.
\end{figure}

\appendix
\justifying

\section{Functional analysis}
\label{spectral}

\subsection{Problem statement}

Let $\Omega = B(0, L) \subset \R^3$ be a bounded domain. We partition $\Omega$ into an inner subdomain $\Oin = B(0, R)$ and an outer annulus $\Oout = \Omega \setminus \overline{\Oin}$, separated by the interface $\GammaInt = \partial \Oin$. We consider the concentration field $c(t, x)$ governed by the following initial-boundary value problem:
\begin{small}
\begin{subequations}
\begin{align}
    \frac{\partial c}{\partial t} - \nabla \cdot (D(x) \nabla c) &= 0 & \text{in } \Omega \times (0, T), \label{eq:pde} \\
    -\Din \partial_r c_{\mathrm{in}} &= -\Dout \partial_r c_{\mathrm{out}} & \text{on } \GammaInt, \label{eq:flux} \\
    -\Din \partial_r c_{\mathrm{in}} &= \kappa (c_{\mathrm{in}} - \Gamma c_{\mathrm{out}}) & \text{on } \GammaInt, \label{eq:jump} \\
    \Dout \partial_r c_{\mathrm{out}} &= 0 & \text{on } \partial \Omega, \label{eq:neumann} \\
    c(0, \cdot) &= c_0 & \text{in } \Omega,
\end{align}
\end{subequations}
\end{small}
where $D(x)$ takes the values $\Din$ in $\Oin$ and $\Dout$ in $\Oout$, $\Gamma>0$ is the partition coefficient and $\kappa>0$ the interfacial permeability. We use the spherical coordinates $x=r\omega$, $r=|x|$, $\omega\in S^2$, with the radial unit vector $\vec e_r=\omega$. The radial derivative is $\partial_r=\vec e_r\cdot\nabla$, so that $-\Din\partial_r c_{\mathrm{in}}$ is the outward flux across $\GammaInt$; the volume element is $dx=r^2\,dr\,d\omega$, $d\omega$ being the area element of the unit sphere, and $d\sigma$ denotes the area element of $\GammaInt$.

\subsection{Functional framework}

The standard Sobolev space $H^1(\Omega)$ is ill-suited for this problem due to the discontinuity of the solution at the interface $r=R$ induced by $\Gamma \neq 1$. Instead, we work with a ``broken'' space. To address the asymmetry introduced by $\Gamma$, we define specific Hilbert spaces equipped with weighted inner products.

\begin{definition}[Hilbert spaces and inner products]
\leavevmode
\begin{enumerate}
    \item Let $H = L^2(\Oin) \times L^2(\Oout)$ be the pivot space. We equip $H$ with the weighted inner product:
    \begin{equation}
    \label{eq:weighted-inner-product}
        \innerH{u}{v} = \int_{\Oin} u_{\mathrm{in}} v_{\mathrm{in}} \, dx + \Gamma \int_{\Oout} u_{\mathrm{out}} v_{\mathrm{out}} \, dx.
    \end{equation}
    The associated norm is denoted by $|u|_H = \sqrt{\innerH{u}{u}}$. Since $\Gamma > 0$, this norm is equivalent to the standard $L^2(\Omega)$ norm.

    \item Let $V = H^1(\Oin) \times H^1(\Oout)$. We equip $V$ with the weighted inner product:
    \begin{equation}
    \begin{split}
        \innerV{u}{v} = & \int_{\Oin} (u_{\mathrm{in}} v_{\mathrm{in}} + \nabla u_{\mathrm{in}} \cdot \nabla v_{\mathrm{in}}) \, dx \\
        & + \Gamma \int_{\Oout} (u_{\mathrm{out}} v_{\mathrm{out}} + \nabla u_{\mathrm{out}} \cdot \nabla v_{\mathrm{out}}) \, dx.
    \end{split}
    \end{equation}
    The associated norm is denoted by $\norm{u}_V = \sqrt{\innerV{u}{u}}$.
\end{enumerate}
\end{definition}

\begin{lemma}[Gelfand triple]
The spaces $V$ and $H$ satisfy the dense and continuous inclusions:
\begin{equation}
    V \hookrightarrow H \equiv H' \hookrightarrow V'.
\end{equation}
\end{lemma}

\begin{proof}
The algebraic inclusion $V \subset H$ is evident. By expanding the inner product in $V$, we observe that:
\[ \norm{u}_V^2 = |u|_H^2 + \left( \|\nabla u_{\mathrm{in}}\|_{L^2}^2 + \Gamma \|\nabla u_{\mathrm{out}}\|_{L^2}^2 \right). \]
Thus, $|u|_H \le \norm{u}_V$, proving the continuity of the injection with constant $1$.
Density follows from the density of $C_c^\infty(\Oin) \times C_c^\infty(\Oout)$ in $H$, noting that these test functions are elements of $V$.
Identifying $H$ with its dual $H'$ via the Riesz representation theorem using the weighted inner product $\innerH{\cdot}{\cdot}$, we obtain the canonical inclusion $H \hookrightarrow V'$.
\end{proof}

\subsection{Variational formulation}

Let $c(t) = (c_{\mathrm{in}}(t), c_{\mathrm{out}}(t))$ be the trial function and let $v = (v_{\mathrm{in}}, v_{\mathrm{out}}) \in V$ be a time-independent test function. To obtain a symmetric formulation compatible with the jump condition, we multiply the governing equation \eqref{eq:pde} on $\Oin$ by $v_{\mathrm{in}}$ and on $\Oout$ by $\Gamma v_{\mathrm{out}}$. Summing the integrals over the respective domains yields:
\begin{equation}
\begin{split}
    & \int_{\Oin} \partial_t c_{\mathrm{in}} v_{\mathrm{in}} \, dx + \Gamma \int_{\Oout} \partial_t c_{\mathrm{out}} v_{\mathrm{out}} \, dx \\
    = & \int_{\Oin} \nabla \cdot (\Din \nabla c_{\mathrm{in}}) v_{\mathrm{in}} \, dx \\
    & + \Gamma \int_{\Oout} \nabla \cdot (\Dout \nabla c_{\mathrm{out}}) v_{\mathrm{out}} \, dx.
\end{split}
\end{equation}
The left-hand side is precisely the time derivative of the weighted inner product, $\frac{d}{dt} (c, v)_H$.

We integrate the right-hand side by parts in each subdomain. On $\GammaInt$ the outward normal is $\vec e_r$ for $\Oin$ and $-\vec e_r$ for $\Oout$, and the Neumann condition \eqref{eq:neumann} cancels the term on $\partial\Omega$:
\begin{equation*}
\begin{split}
    \text{RHS} =& -\int_{\Oin}\Din\nabla c_{\mathrm{in}}\cdot\nabla v_{\mathrm{in}}\,dx
                \\&  -\Gamma\int_{\Oout}\Dout\nabla c_{\mathrm{out}}\cdot\nabla v_{\mathrm{out}}\,dx\\
               & + \int_{\GammaInt}\bigl(\Din\partial_r c_{\mathrm{in}}\,v_{\mathrm{in}}
                  -\Gamma\Dout\partial_r c_{\mathrm{out}}\,v_{\mathrm{out}}\bigr)\,d\sigma .
\end{split}
\end{equation*}
With the flux $J=-\Din\partial_r c_{\mathrm{in}}=-\Dout\partial_r c_{\mathrm{out}}$ of Eq.~\eqref{eq:flux}, the interface term is $\int_{\GammaInt}J\,(\Gamma v_{\mathrm{out}}-v_{\mathrm{in}})\,d\sigma$, and the jump condition \eqref{eq:jump}, $J=\kappa(c_{\mathrm{in}}-\Gamma c_{\mathrm{out}})$, turns it into $-\kappa\int_{\GammaInt}(c_{\mathrm{in}}-\Gamma c_{\mathrm{out}})(v_{\mathrm{in}}-\Gamma v_{\mathrm{out}})\,d\sigma$.

Based on this derivation, we define the symmetric bilinear form $a: V \times V \to \R$ as:
\begin{equation}
\begin{split}
    a(u, v) = & \int_{\Oin} \Din \nabla u_{\mathrm{in}} \cdot \nabla v_{\mathrm{in}} \, dx \\
    &+ \Gamma \int_{\Oout} \Dout \nabla u_{\mathrm{out}} \cdot \nabla v_{\mathrm{out}} \, dx \\
    &+ \kappa\int_{\GammaInt} (u_{\mathrm{in}} - \Gamma u_{\mathrm{out}})(v_{\mathrm{in}} - \Gamma v_{\mathrm{out}}) \, d\sigma.
\end{split}
\end{equation}

The variational problem is: Find $c \in L^2(0,T; V) \cap C([0,T]; H)$ with $\partial_t c \in L^2(0,T; V')$ such that:
\begin{equation} \label{eq:weak_final}
    \frac{d}{dt} (c(t), v)_{H} + a(c(t), v) = 0 \quad \forall v \in V,
\end{equation}
subject to the initial condition $c(0) = c_0$.

\subsection{Existence and uniqueness \label{sec:existence}}

We use the existence and uniqueness theorem of J.-L. Lions for parabolic problems~\cite{brezis2011functional}.

\begin{lemma}[Properties of the bilinear form] \label{lem:prop}
The bilinear form $a(\cdot, \cdot)$ satisfies:
\begin{enumerate}[label=(\roman*)]
    \item \textbf{Continuity:} There exists $M > 0$ such that $|a(u, v)| \le M \norm{u}_V \norm{v}_V$ for all $u, v \in V$.
    \item \textbf{Coercivity (Gårding's inequality):} There exist $\alpha > 0$ and $C_0 \ge 0$ such that $a(u, u) \ge \alpha \norm{u}_V^2 - C_0 |u|_H^2$ for all $u \in V$.
\end{enumerate}
\end{lemma}

\begin{proof}
We decompose the bilinear form into volume terms ($I_{\text{vol}}$) and the boundary term ($I_{\text{bnd}}$).

\noindent \textbf{(i) Proof of continuity:} We estimate each term separately.
Using the Cauchy--Schwarz inequality and the definition of the norm induced by $\innerV{\cdot}{\cdot}$:
\begin{align*}
    & \quad |I_{\text{vol}}| \\
    &= \left| \int_{\Oin} \Din \nabla u_{\mathrm{in}} \cdot \nabla v_{\mathrm{in}} + \Gamma \int_{\Oout} \Dout \nabla u_{\mathrm{out}} \cdot \nabla v_{\mathrm{out}} \right| \\
    &\le D_{\max} \left( \|\nabla u_{\mathrm{in}}\|_{L^2} \|\nabla v_{\mathrm{in}}\|_{L^2} + \Gamma \|\nabla u_{\mathrm{out}}\|_{L^2} \|\nabla v_{\mathrm{out}}\|_{L^2} \right) \\
    &\le D_{\max} \norm{u}_V \norm{v}_V,
\end{align*}
with $D_{\max} = \max(\Din, \Dout)$.
For $I_{\text{bnd}}$, the trace theorem implies there exists $C_{\mathrm{tr}}$ such that $\|w\|_{L^2(\GammaInt)} \le C_{\mathrm{tr}} \|w\|_{H^1}$. By the triangle inequality:
\[
    \|u_{\mathrm{in}} - \Gamma u_{\mathrm{out}}\|_{L^2(\GammaInt)} \le C_{\mathrm{tr}} \|u_{\mathrm{in}}\|_{H^1} + \Gamma C_{\mathrm{tr}} \|u_{\mathrm{out}}\|_{H^1}.
\]
Given the definition of $\norm{u}_V$, there exists $C_\Gamma$ such that $\|u_{\mathrm{in}} - \Gamma u_{\mathrm{out}}\|_{L^2(\GammaInt)} \le C_\Gamma \norm{u}_V$. Thus:
\[ |I_{\text{bnd}}| \le \kappa C_\Gamma^2 \norm{u}_V \norm{v}_V. \]
Summing these bounds, continuity holds.

\bigskip

\noindent \textbf{(ii) Proof of coercivity:}
Since $\kappa > 0$, the boundary term $I_{\text{bnd}}(u,u)$ is nonnegative. Let $D_{\min} = \min(\Din, \Dout)$. We have:
\[ a(u, u) \ge D_{\min} \left( \|\nabla u_{\mathrm{in}}\|_{L^2}^2 + \Gamma \|\nabla u_{\mathrm{out}}\|_{L^2}^2 \right). \]
By the definition of the inner products above,
\[ \innerV{u}{u} = \innerH{u}{u} + \left( \|\nabla u_{\mathrm{in}}\|_{L^2}^2 + \Gamma \|\nabla u_{\mathrm{out}}\|_{L^2}^2 \right). \]
Therefore:
\[ \|\nabla u_{\mathrm{in}}\|_{L^2}^2 + \Gamma \|\nabla u_{\mathrm{out}}\|_{L^2}^2 = \norm{u}_V^2 - |u|_H^2. \]
Substituting this relation into the inequality for $a(u,u)$:
\[ a(u, u) \ge D_{\min} (\norm{u}_V^2 - |u|_H^2). \]
This yields Gårding's inequality with $\alpha = D_{\min}$ and $C_0 = D_{\min}$.
\end{proof}

\begin{theorem}[Existence and uniqueness]
Given $c_0 \in H$, there exists a unique solution $c$ to the problem \eqref{eq:weak_final} satisfying:
$$ c \in L^2(0, T; V) \cap C([0, T]; H) \quad \text{and} \quad \frac{d c}{dt} \in L^2(0, T; V'). $$
\end{theorem}

\begin{proof}
The change of unknown $\tilde c = e^{-C_0 t} c$ turns Eq.~\eqref{eq:weak_final} into the same problem for the form $a + C_0 \innerH{\cdot}{\cdot}$, which is coercive on $V$ by Lemma~\ref{lem:prop}. The result then follows from Lions' theorem. The measurability of $t \mapsto a(u,v)$ is trivial as the coefficients are time-independent.
\end{proof}

\subsection{Spectral analysis}

Let $\mathcal A$ be the operator associated with $a$ on $H$, $\innerH{\mathcal Au}{v}=a(u,v)$ for all $v\in V$, defined on the $u\in V$ for which $v\mapsto a(u,v)$ is continuous for the norm of $H$. For piecewise $H^2$ functions, this domain is characterized by the interface and boundary conditions \eqref{eq:flux}--\eqref{eq:neumann}; continuity at the interface is not imposed.

\begin{theorem}[Spectral decomposition\label{thm:spectral}]
The operator $\mathcal{A}$ is self-adjoint and admits a purely discrete spectrum consisting of a sequence of real eigenvalues:
\[ 0 = \lambda_0 < \lambda_1 \le \dots \le \lambda_k \le \dots \to +\infty. \]
The eigenvalue $\lambda_0 = 0$ is simple. Moreover, the associated eigenvectors $\{\Phi_k\}_{k \in \mathbb{N}}$ form a complete orthonormal basis of the weighted space $H$.
\end{theorem}

\begin{proof}
The form $a$ is symmetric, so $\mathcal A$ is self-adjoint. By Gårding's inequality (Lemma~\ref{lem:prop}), the form $a+\mu\innerH{\cdot}{\cdot}$ is coercive on $V$ for $\mu>C_0$, and the Lax--Milgram theorem makes the resolvent $(\mathcal A+\mu I)^{-1}$ bounded from $H$ to $V$. Since $\Oin$ and $\Oout$ are bounded Lipschitz domains, the embedding $V\hookrightarrow H$ is compact (Rellich--Kondrachov), so that the resolvent is a compact self-adjoint operator on $H$. The spectral theorem for such operators~\cite{brezis2011functional} gives an orthonormal basis of eigenfunctions, with eigenvalues $\lambda_k\to+\infty$, nonnegative since $a(u,u)\ge0$. Finally, $a(u,u)=0$ forces $u$ to be constant in each subdomain with $u_{\mathrm{in}}=\Gamma u_{\mathrm{out}}$: the kernel is spanned by $h=(\Gamma,1)$, and $\lambda_0=0$ is simple.
\end{proof}

With $V_{\mathrm{in}}=|\Oin|$ and $V_{\mathrm{out}}=|\Oout|$, the normalized equilibrium mode is
\begin{equation}
    \Phi_0=\frac{h}{\sqrt{\Gamma^2V_{\mathrm{in}}+\Gamma V_{\mathrm{out}}}},\qquad
    h=\begin{cases}\Gamma&\text{in }\Oin,\\ 1&\text{in }\Oout.\end{cases}
    \label{eq:zero-mode}
\end{equation}
Taking $v=h$ in Eq.~\eqref{eq:weak_final}, for which $a(c,h)=0$, shows that the total amount is conserved,
\begin{equation}
    M_0=\int_\Omega c(t,x)\,dx=\Gamma^{-1}\innerH{c(t)}{h}=\int_\Omega c_0(x)\,dx,
    \label{eq:mass}
\end{equation}
and the solution relaxes to $c_\infty=M_0\,h/(\Gamma V_{\mathrm{in}}+V_{\mathrm{out}})$, the partition $c_{\mathrm{in}}=\Gamma c_{\mathrm{out}}$.

\subsection{Separation of the angular variables}

The gradient splits as $\nabla=\vec e_r\partial_r+r^{-1}\nabla_{S^2}$, and the angular part involves the Laplace--Beltrami operator $-\Delta_{S^2}$, self-adjoint on $L^2(S^2)$ with eigenvalues $\ell(\ell+1)$, $\ell\in\mathbb N$. Its eigenfunctions, the restrictions to $S^2$ of the homogeneous harmonic polynomials of degree $\ell$, define the real spherical harmonics $Y_{\ell m}$, $|m|\le\ell$, an orthonormal basis of $L^2(S^2)$. Since $D$ depends on $r$ only and $\GammaInt$ and $\partial\Omega$ are spheres, the $Y_{\ell m}$ decouple the weak problem; with $d\sigma=R^2\,d\omega$ on $\GammaInt$, $\innerH{fY_{\ell m}}{gY_{\ell'm'}}=\delta_{\ell\ell'}\delta_{mm'}(f,g)_r$ and $a(fY_{\ell m},gY_{\ell'm'})=\delta_{\ell\ell'}\delta_{mm'}\,a_\ell(f,g)$, with $(f,g)_r=\int_0^R fg\,r^2dr+\Gamma\int_R^L fg\,r^2dr$ and
\begin{equation}
\begin{split}
    a_\ell(f,g)=&\int_0^R \Din\Bigl(f'g'+\tfrac{\ell(\ell+1)}{r^2}fg\Bigr)r^2dr\\
    &+\Gamma\int_R^L \Dout\Bigl(f'g'+\tfrac{\ell(\ell+1)}{r^2}fg\Bigr)r^2dr\\
    &+\kappa R^2\bigl[f_{\mathrm{in}}(R)-\Gamma f_{\mathrm{out}}(R)\bigr]\\
    &\quad\times\bigl[g_{\mathrm{in}}(R)-\Gamma g_{\mathrm{out}}(R)\bigr].
\end{split}
\label{eq:radial-form}
\end{equation}
The problem thus splits into one radial problem per degree $\ell$, independent of $m$: the eigenfunctions of Theorem~\ref{thm:spectral} are $\Phi_{\ell mn}=f_{\ell n}(r)\,Y_{\ell m}(\omega)$, with $f_{\ell n}$ the radial eigenfunctions normalized for $(\cdot,\cdot)_r$ and eigenvalues $\lambda_{\ell n}$ of multiplicity $2\ell+1$.

\subsection{Green's function representation}

Using the spectral decomposition of $\mathcal{A}$, the solution of the variational problem \eqref{eq:weak_final} for $c_0\in H$ is
\begin{equation}
    c(t) = \sum_{k} e^{-\lambda_k t}\, \innerH{c_0}{\Phi_k}\, \Phi_k,
\end{equation}
where the single index $k$ runs over the equilibrium mode $\Phi_0$ and all modes $(\ell,m,n)$. With the weight $w=1$ in $\Oin$ and $w=\Gamma$ in $\Oout$, this reads $c(t,x)=\int_\Omega G_t(x,x_0)\,c_0(x_0)\,w(x_0)\,dx_0$, with the kernel of Eq.~\eqref{eq:green-function},
\begin{equation}
    G_t(x, x_0) = \sum_{k} e^{-\lambda_k t}\, \Phi_k(x)\, \Phi_k(x_0).
\end{equation}
A molecule initially at $x_0=(r_0,\theta_0,\phi_0)$ is a Dirac mass $\delta_{x_0}$, which is not an element of $H$; its evolution is defined for $t>0$ through this kernel, the series converging thanks to the factors $e^{-\lambda_k t}$. Its expansion coefficients are
\begin{align}
    c_{\ell m n} & = \innerH{\delta_{x_0}}{\Phi_{\ell m n}} \notag\\
                & =
                \begin{cases}
                    \Phi_{\ell m n}(r_0, \theta_0, \phi_0) & \textrm{for } x_0\in \Oin, \\
                    \Gamma\, \Phi_{\ell m n}(r_0, \theta_0, \phi_0) & \textrm{for } x_0\in \Oout,
                \end{cases}
                \label{eq:coefficients}
\end{align}
and the transition density with respect to the ordinary volume is
\begin{equation}
\begin{split}
    p_t(x | x_0) & = w(x_0)\, G_t(x, x_0) \\
    & = \sum_{k} c_k(x_0)\, e^{-\lambda_k t}\, \Phi_k(x).
\end{split}
\label{eq:distribution}
\end{equation}
From a physical perspective, $p_t(x|x_0)$ is the probability density for a particle initially located at $x_0$ to be found at $x$ at time $t$: it is nonnegative and integrates to one by Eq.~\eqref{eq:mass}. For $x_0\in\Oin$, $p_t=G_t$. At long times, $p_t$ tends to the equilibrium distribution $h(x)/(\Gamma V_{\mathrm{in}}+V_{\mathrm{out}})$, independent of $x_0$.

\subsection{Radial basis}

For $\lambda>0$, the radial eigenfunctions are spherical Bessel functions, with eigenvalues given by a transcendental equation~\cite{grebenkov2008analytical,carr2018modelling,Moutal2019}. Instead of solving it for each $\ell$, we use a Galerkin basis adapted to their form. Inside, the eigenfunction regular at the centre, $j_\ell(r\sqrt{\lambda/\Din})$, is $r^\ell$ times a smooth function of $r^2$, as is the radial factor of any function smooth at the centre. We therefore use the radial Zernike functions $Z_{\ell n}(r)\propto r^\ell P_n^{(0,\ell+1/2)}(2r^2/R^2-1)$, $0\le n<N_{\mathrm{in}}$, where the Jacobi polynomials $P_n^{(0,\ell+1/2)}$ make them orthonormal for $r^2dr$ on $[0,R]$. Outside, where the radial factors vary on a scale proportional to $r$, we use Legendre polynomials in $\log r$, orthonormalized for $r^2dr$ on $[R,L]$. Few functions per degree suffice, as the truncation check below shows.

\subsection{Symmetric discretization and time evolution}
For numerical stability, we evaluate the video model by diagonalizing a symmetric Galerkin discretization of the transport operator. Write the concentration as $c_N(x,t)=\sum_i a_i(t)\phi_i(x)$ in a finite basis orthonormal for $\innerH{\cdot}{\cdot}$. The coefficients satisfy
\begin{equation}
    \dot{\mathbf a}=-B\mathbf a,\qquad B_{ij}=a(\phi_j,\phi_i).
    \label{eq:discrete-generator}
\end{equation}
The symmetric matrix $B$ is diagonalized as $B=Q\Lambda Q^\mathsf{T}$ using a Hermitian eigensolver (\texttt{eigh}). Its time evolution is then
\begin{equation}
    \mathbf a(t)=Qe^{-\Lambda t}Q^\mathsf{T}\mathbf a(0).
    \label{eq:discrete-propagation}
\end{equation}
Thus, the analytical exponential evolution is retained for the discretized system, while its rates and modes are obtained by matrix diagonalization.

In the implementation, the basis functions are the products of $Y_{\ell m}$ with $Z_{\ell n}$ inside and with the outer radial functions outside. The inner and outer coefficients are independent, allowing a concentration jump at the interface, and scaling the outer coefficients by $\sqrt{\Gamma}$ makes the weighted inner product Euclidean: the basis being orthonormal, there is no mass matrix and $B$ is symmetric. By Eq.~\eqref{eq:radial-form}, $B$ is block diagonal, with one radial block per degree $\ell$ shared by its azimuthal components. The polar axis of the spherical harmonics is taken normal to the boundary of the half-bleach region, so that the bleached state is close to axisymmetric: the degree $\ell$ resolves the bleach boundary, and a few azimuthal orders $|m|$ suffice for the departures from axial symmetry (finite bleach region, depth dependence of the dose). Video fits retain $\ell\leq12$, $|m|\leq4$ and 12 inner radial functions, with 24 outer radial functions for PLL-HA ($L=32R$) and 16 for DDX4 ($L=8R$). We checked the truncation: increasing it to $\ell\le20$ with 20 inner radial functions, or to $|m|\le8$, changed $\Din$ and $\kappa$ by less than 0.4\% at 0 and 140~mM MgCl\textsubscript{2}. The observation operator is applied to the propagated state to produce confocal images. For the measured-reservoir treatment of DDX4-YFP, the prescribed external fluorescence supplies a time-dependent driving term, integrated with the same spectral propagator.

\section{Idealized half- and full-FRAP signals \label{frap-predictions}}
For instantaneous bleaching of a hemisphere $\Omega_\mathrm{bleach}$, the normalized initial bleach profile is given by $c_0(x)=\mathbf{1}_{\Omega_\mathrm{bleach}}(x)/V_\mathrm{bleach}$. Its evolution is
\begin{equation}
    c^\mathrm{half}(x,t)=\frac{1}{V_\mathrm{bleach}}
    \int_{\Omega_\mathrm{bleach}}G_t(x,x')\,dx'.
    \label{eq:integrated-curves}
\end{equation}
Because the initial bleach profile is nonzero only within the condensate, $G_t$ is the corresponding volume transition density. Integrating over the two hemispheres gives
\begin{equation}
    c_{\mathrm{half},q}(t)=\int_{\Omega_q}c^\mathrm{half}(x,t)\,dx,
    \qquad q\in\{\mathrm{b},\mathrm{nb}\},
    \label{eq:integrated-curves2}
\end{equation}
with $\mathrm{FRAP}_{\mathrm{half},q}=1-c_{\mathrm{half},q}$. Initially these fluorescence signals equal zero and one, respectively. Volume integration and this normalization yields the expressions for the idealized model considered in Section~\ref{sec:idealized-frap}.

For the corresponding uniform bleach of the entire sphere,
\begin{equation}
    c^\mathrm{full}(x,t)=\frac{1}{V_\mathrm{in}}
    \int_{\Oin}G_t(x,x')\,dx',
    \label{eq:integrated-curves-full}
\end{equation}
and
\begin{equation}
    c_\mathrm{full}(t)=\int_{\Oin}c^\mathrm{full}(x,t)\,dx,
    \label{eq:integrated-curves3}
\end{equation}
\begin{equation}
    \mathrm{FRAP}_\mathrm{full}(t)=1-c_\mathrm{full}(t).
    \label{eq:integrated-curves4}
\end{equation}
The initial profile for bleaching the entire sphere is the average of the profiles for bleaching the two complementary hemispheres. By linearity, its response is their average. Rotational symmetry makes their whole-condensate bleach profiles equal, so
\begin{equation}
    c_\mathrm{full}(t)=c_\mathrm{half,b}(t)+c_\mathrm{half,nb}(t).
\end{equation}
This yields Eq.~\eqref{eq:half-full-sum}. The modes selected by the sum and by the difference of the half signals are discussed in Section~\ref{sec:idealized-frap}. These observables can be calculated by integrating the propagated field. Their evaluation does not require explicit radial eigenfunctions.

The examples in Fig.~\ref{half-frap-parameter-influence} illustrate the idealized recovery curves while one parameter is varied with respect to the case of free diffusion ($\rho=0$, $\Gamma=1$, $\Delta=1$, $\kappa\to\infty$). They provide the physical comparisons underlying the dip--$M$ relation in the main text (Fig.~\ref{dip-dependence}). For free diffusion in an unbounded medium, the full-FRAP curve has the closed form given in Appendix~\ref{free-diffusion}. 
\begin{figure}
\centering
\begin{tikzpicture}
\node[anchor=south west,inner sep=0] (image) at (0,0) {\includegraphics[width=\linewidth]{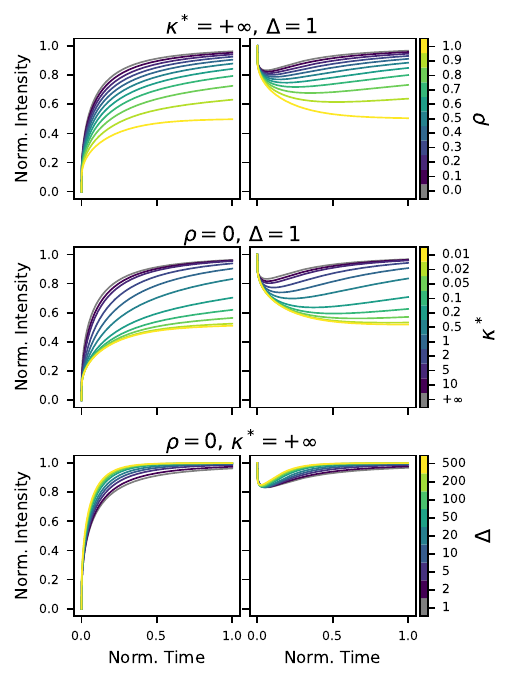}};
\begin{scope}[x={(image.south east)},y={(image.north west)}]
    \node[anchor=west] at (0.0,0.96) {\textbf{(a)}};
    \node[anchor=west] at (0.0,0.66) {\textbf{(b)}};
    \node[anchor=west] at (0.0,0.36) {\textbf{(c)}};
\end{scope}
\end{tikzpicture}
\caption{Integrated recovery curves for an idealized half-FRAP experiment.}
\label{half-frap-parameter-influence}
\justifying \noindent
(a) Increasing directional reflective bias of the interface $\rho$ slows down the recovery in the bleached half and increases the dip depth in the non-bleached half. (b) Increasing dimensionless interfacial resistance $1/\kappa^*$ has similar effects as increasing directional reflective bias. (c) The ratio of diffusion coefficients $\Delta$ has only milder effects, with increasing $\Delta$ leading to a faster recovery to the equilibrium state. All curves were plotted for $L/R=100$.
\end{figure}

\section{Reflective bias of the interface \label{reflective-bias}}
The interface law can be written as a difference of directional contributions,
\begin{equation}
    J_\text{net}=\kin\cin-\kout\sqrt{\Delta}\,\cout,
    \label{eq:flux-decomposition}
\end{equation}
where $\kin=\kappa$ and $\kout=\kappa\Gamma/\sqrt{\Delta}$. We define normalized directional weights by
\begin{align}
    \rhoin&\eqdef\frac{\kout}{\kout+\kin}
      =\frac{\Gamma/\sqrt{\Delta}}{1+\Gamma/\sqrt{\Delta}},\notag\\
    \rhoout&\eqdef\frac{\kin}{\kout+\kin}
      =\frac{1}{1+\Gamma/\sqrt{\Delta}}.
    \label{eq:reflection-coefficients}
\end{align}
These nonnegative weights sum to one. Their difference is the directional reflective bias given in Eq.~\eqref{eq:rho-definition}, $\rho=\rhoin-\rhoout=(\kout-\kin)/(\kout+\kin)$.
For positive finite $\Gamma$ and $\Delta$, $-1<\rho<1$, with $\rho=0$ when $\Gamma=\sqrt{\Delta}$. The endpoint limits describe vanishing and diverging ratios of the rescaled coefficients and do not fix either coefficient's absolute magnitude. This algebraic parametrization separates directional asymmetry from the permeability $\kappa$. In the one-dimensional sharp-interface model without additional interfacial resistance, the same expression governs reflection of the diffusive flux~\cite{bo2021stochastic}. Because the flux law is a continuum description and does not resolve individual encounters with the interface, the weights $\rhoin$ and $\rhoout$ should not be interpreted as microscopic reflection probabilities, but rather as normalized directional coefficients $\kout$ and $\kin$, which quantify the directional asymmetry of the interface but not how fast molecules cross it. The inverse relation is
\begin{equation}
    \Gamma=\frac{1+\rho}{1-\rho}\sqrt{\Delta}.
\end{equation}

\section{Escape from the inner domain \label{escape-time}}
For this calculation, we replace the reflecting outer boundary of the FRAP model by an absorbing boundary at $r=R+\delta$. Let $p_t(x,x_0)$ denote the transition density of this killed process. The accumulated internal residence before absorption is
\begin{equation}
    T(x_0) = \int_{t=0}^{\infty} \int_{x \in \Oin} p_t(x, x_0)\, dx\, dt,
\end{equation}
Unlike the transition density in Eq.~\eqref{eq:distribution} derived for a reflecting boundary at $r=L$, this density loses total mass at the absorbing boundary. The spatial integral defining $T$ counts only time spent in $\Oin$, including repeated visits. The same weighted reciprocity relation holds for $x\in\Oin$, since the weighted symmetry of Appendix~\ref{spectral} is preserved by the Dirichlet condition at $r=R+\delta$:
\begin{alignat}{3}
    p_t(x, x_0) &= p_t(x_0, x) && \text{if } x_0 \in \Oin, \notag\\
    p_t(x, x_0) &= \Gamma p_t(x_0, x)\,\,\, && \text{if } x_0 \in \Oout.
\end{alignat}
Next, we define $\Tin$ and $\Tlayer$, the mean residence times in $\Oin$ for molecules initially located in the inner domain and the boundary layer of size $\delta$, respectively:
\begin{align}
    \Tin(x_0) & = \int_{t=0}^{\infty} \int_{x \in \Oin} p_t(x_0, x) \,dx\,dt, \notag\\
    \Tlayer(x_0) & = \Gamma \int_{t=0}^{\infty} \int_{x \in \Oin} p_t(x_0, x) \,dx\,dt.
\end{align}
For a fixed $x$, we define $c(t, x_0)\eqdef p_t(x_0, x)$, which is the probability density for a particle to be found at position $x_0$ at time $t$ given that its initial position was $x$. The function $c(t, x_0)$ satisfies the diffusion equations
\begin{alignat}{3}
    \frac{\partial  c_{\mathrm{in}}}{\partial t} &= \Din \Delta c_{\mathrm{in}} && \text{for } 0 \leq r_0 \leq R, \notag\\
    \frac{\partial  c_{\mathrm{layer}}}{\partial t} &= \Dout \Delta c_{\mathrm{layer}} \,\,\,\,\,\,&& \text{for } R \leq r_0 \leq R+\delta,
    \label{eq:diffusion_eq2}
\end{alignat}
with the initial condition $c(0,x_0)=\delta(x_0-x)$. In Eq. \eqref{eq:diffusion_eq2}, $r_0 = |x_0|$ denotes the radial coordinate. To obtain the equations for the residence time $T$, we integrate the diffusion equation \eqref{eq:diffusion_eq2} in space and time. On the LHS, the following equation is obtained:
\begin{alignat}{3}
    \int_{t=0}^{\infty} &\int_{x \in \Oin} \frac{\partial c}{\partial t} \,dx\, dt &&= \int_{x \in \Oin} c(\infty,x_0) \,dx \notag\\ - &\int_{x \in \Oin} c(0,x_0) \,dx &&= \int_{x \in \Oin} p_{t=\infty}(x_0, x) \,dx \notag\\ - &\int_{x \in \Oin} p_{t=0}(x_0,x) \,dx 
     &&= 0 - \begin{cases}
        1 & \text{for } 0 \leq r_0 \leq R \\
        0 & \text{for } R \leq r_0 \leq R + \delta
    \end{cases}
\end{alignat}
On the RHS, the operator $\int_{t=0}^{\infty} \int_{x \in \Oin} (\cdot) \,dx\, dt$ commutes with the spatial derivatives:
\begin{alignat}{3}
    \Din \Delta T &= -1 \,\,\,\,\,&& \text{for } 0 \leq r_0 \leq R, \notag\\
    \Dout \Delta T &= 0 && \text{for } R \leq r_0 \leq R + \delta.
\end{alignat}
They follow from the interface conditions of the forward model, with $\cin$ replaced by $\Tin$ and $\cout$ by $\Tlayer/\Gamma$:
\begin{align}
    -\Din \Gamma \partial_r \Tin|_{r=R} &= -\Dout \partial_r \Tlayer|_{r=R}, \notag\\
    -\Din \partial_r  \Tin|_{r=R} &= \kappa ( \Tin(R_-) - \Tlayer(R_+) ), \notag\\
    \Tlayer(R+\delta) &= 0.
\end{align}
The first two equations describe the semipermeable interface that conserves flux. The third equation is a Dirichlet condition, which is applied at the edge of the boundary layer, $r=R+\delta$, to treat molecules at this location as having escaped. We obtain the solution of the Laplace equation:
\begin{align}
    \Tin(r_0) & = -\frac{r_0^2}{6\Din}+K_1, \notag\\
    \Tlayer(r_0) & = \frac{K_2}{r_0}+K_3,
\end{align}
where $K_1$, $K_2$ and $K_3$ are three constants that are determined by the boundary conditions and read:
\begin{align}
    K_1 & = \frac{R^2}{6 D_{\rm in}} + \frac{R}{3 \kappa} + \frac{\Gamma R^2}{3 D_{\rm out}}\left(\frac{\delta}{R+\delta} \right), \notag\\
    K_2 & = \frac{\Gamma R^3}{3 D_{\rm out}}, \,\,
    K_3 = -\frac{\Gamma R^3}{3 D_{\rm out} (R+\delta)} .
\end{align}
The formula above is valid for any specified $\delta>0$. To calculate the degree of internal mixing used in the main text, we choose $\delta$ from a characteristic outer diffusion length over the reference timescale $R^2/(\pi^2\Din)$~\cite{zhang2024the}:
\begin{equation}
    \delta = \sqrt{\Dout R^2/(\pi^2\Din)} = R/\pi \sqrt{\Dout/\Din}.
\end{equation}
This is a convention for the escape surface, rather than a measured interface width or an exact hemispherical mixing time. Defining the effective diffusion coefficient according to
\begin{equation}
    \Dout^\text{eff} \eqdef D_{\rm out} \left(\frac{R+\delta}{\delta} \right) = \Dout + \sqrt{\pi^2 \Din \Dout},
\end{equation}   
the mean residence time for a particle starting at the center of the inner domain, $r_0=0$, reads
\begin{equation}
    \Tin(0) = K_1 = \frac{R^2}{6 D_{\rm in}} + \frac{R}{3 \kappa} + \frac{\Gamma R^2}{3 \Dout^\text{eff}}.
    \label{eq:tau-escape}
\end{equation}
The first term is the mean first-passage time from the center to the condensate interface. The second and third terms are additional residence times within the condensate due to finite permeability and returns from the outer layer, respectively. Averaging the initial position uniformly over the inner sphere would replace the first term by $R^2/(15\Din)$, leaving the other terms unchanged. Computing total elapsed time to absorption would also count time spent outside and would require a different backward equation in the outer layer. For $\Dout\gg\Din$, our convention gives $\delta\gg R$ and $\Dout^\text{eff}\to\Dout$.

For a free-diffusion reference with an escape surface at $R+\delta_\text{free}$, Eq.~\eqref{eq:tau-escape} gives
\begin{equation}
    \tau_\text{free}=\frac{R^2}{6\Din}
    \left(1+\frac{2\delta_\text{free}}{R+\delta_\text{free}}\right).
\end{equation}
Applying the same escape rule to the reference ($\Delta=1$) gives $\delta_\text{free}=R/\pi$ and hence
\begin{equation}
    \tau_\text{free}=\frac{R^2}{6\Din}\left(\frac{3+\pi}{1+\pi}\right).
\end{equation}
With the effective resistance of Eq.~\eqref{eq:effective-resistance}, these two times give the closed form in Eq.~\eqref{eq:internal-mixing}. The factors containing $\pi$ therefore follow from the chosen escape rule. The limits $M\to-1$ and $M\to1$ correspond to vanishing and diverging excess internal residence times, respectively.

\section{Radiation boundary condition \label{limiting-radiation}}
In the limiting case of fast diffusion in an unbounded outer domain, we assume that the solution satisfies the Laplace equation, as the outer domain is always close to its steady state. In this case, we can write the radial part of the eigenfunctions in the outer domain as
\begin{equation}
    f^\text{out}_{\ell}(r) = C_{\ell} \left(\frac{R}{r}\right)^{\ell + 1}
    \label{eq:radial-quasi-eq}
\end{equation}
Here, $C_\ell$ is an amplitude that is not relevant for the following calculation, and the superscripts "in" and "out" denote the two sides of the radial functions $f_{\ell n}$ defined in Appendix~\ref{spectral}. Since $\partial_r f^\text{out}_\ell(R)=-(\ell+1)f^\text{out}_\ell(R)/R$, the effective resistance depends on the angular mode: $\kappa_{\mathrm{eff},\ell}^{-1}=\kappa^{-1}+\Gamma R/[(\ell+1)\Dout]$. For the isotropic modes ($\ell=0$), which are the only modes excited by a uniform bleach of the entire sphere, the boundary conditions read
\begin{align}
   \Din \partial_r f^\text{in}_{0 n}|_{r=R} &= -\frac{\Dout}{R} f^\text{out}_{0 n}(R), \label{eq:robin_boundary_limit}\\
   -\Din \partial_r f^\text{in}_{0 n}|_{r=R} &= \kappa f^\text{in}_{0 n}(R) + \frac{\kappa \Gamma R}{\Dout} \Din \partial_r f^\text{in}_{0 n}|_{r=R}. \label{eq:robin_boundary_limit2}
\end{align}
Eq. \eqref{eq:robin_boundary_limit2} was obtained by inserting $f^\text{out}_{0 n}(R)$ from Eq. \eqref{eq:robin_boundary_limit} into Eq. \eqref{eq:robin_boundary}. Rearranging Eq. \eqref{eq:robin_boundary_limit2} yields
\begin{equation}
        -\Din \partial_r f^\text{in}_{0 n}|_{r=R} = \underbrace{ \left(\frac{1}{\kappa} + \frac{R \Gamma}{\Dout} \right)^{-1}}%
    _{\kappa_{\text{eff,lim}}} f^\text{in}_{0 n}(R).
    \label{eq:radiation-boundary1}
\end{equation}
Eq. \eqref{eq:radiation-boundary1}, which corresponds to Eq.~\eqref{eq:radiation-boundary} in the main text, is the radiation boundary condition we considered recently \cite{muzzopappa2022detecting}, with the effective resistance $\kappa_{\text{eff,lim}}^{-1}$.

\section{Flux balance at the dip minimum \label{integrals}}
For the non-bleached hemisphere, conservation of molecules and the divergence theorem give
\begin{align}
    \partial_t\text{FRAP}_\text{half,nb}(t)
    &=-\partial_t\int_{\Omega_\text{nb}}c\,dx\notag\\
    &=-\int_{\Omega_\text{nb}}\Din\nabla^2c\,dx\notag\\
    &=-\int_{\partial\Omega_\text{nb}}\Din\nabla c\cdot\mathbf n\,d\sigma.
    \label{eq:time-derivative}
\end{align}
Let $\mathbf j=-\Din\nabla c$ and split the boundary into the separating disk $\Sigma_\text{mid}$ and the outer hemispherical surface $\Sigma_\text{interface}$. Then
\begin{align}
    \partial_t\text{FRAP}_\text{half,nb}(t)
    &=\int_{\Sigma_\text{mid}}\mathbf j\cdot\mathbf n\,d\sigma
      +\int_{\Sigma_\text{interface}}\mathbf j\cdot\mathbf n\,d\sigma
      \notag\\
    &=J_\text{mix}+J_\text{interface}.
    \label{eq:time-derivative2}
\end{align}
On the spherical part, $\nabla c\cdot\mathbf n=\partial_r c$ is evaluated on the inner side $r=R_-$. On the separating disk, $\mathbf n$ is the axial unit vector pointing from the non-bleached into the bleached half. Influx of bleached molecules across the separating disk decreases fluorescence in the non-bleached half, whereas efflux across the interface increases it. At the local minimum of $\text{FRAP}_\text{half,nb}$, both contributions balance.

\section{Integrated full-FRAP curve for free diffusion \label{free-diffusion}}
For free diffusion in an unbounded medium, the normalized full-FRAP curve has the closed form
\begin{equation}
    \text{FRAP}_\text{full}^\text{free}(t) = 1 - c_\text{full}^\text{free}\left(\frac{R}{\sqrt{Dt}}\right),
    \label{eq:integrated-curve-free}
\end{equation}
with
\begin{equation}
    c_\text{full}^\text{free}(u) = \text{erf}(u) - \frac{1}{\sqrt\pi} \left[\left(\frac{3}{u} - \frac{2}{u^3} \right) - e^{-u^2} \left(\frac{1}{u}-\frac{2}{u^3}\right)\right].
    \label{eq:integrated-curve-free2}
\end{equation}

This reference corresponds to $\Gamma=1$, $\Din=\Dout=D$, $\kappa\to\infty$, and an unbounded outer phase.

\section{Materials and Methods}
\subsection{Preparation of PLL-HA coacervates}
Coacervates containing Poly-L-Lysine (PLL) and Hyaluronic Acid (HA) were prepared as recently described \cite{muzzopappa2022detecting}. In brief, stock solutions of Poly-L-Lysine hydrobromide (PLL, 15-30 kDa; P7890, Sigma), Poly-L-Lysine-FITC hydrobromide (PLL-FITC, 15-30 kDa; P3543, Sigma), Poly-L-Lysine-Atto647N hydrobromide (prepared from unlabeled PLL and Atto647N, AD 647N-31, ATTO-TEC), and Hyaluronic Acid sodium salt (HA, 8–15 kDa; 40583, Sigma) were prepared in 50 mM Tris-Cl pH 8 at a concentration of 10 mg/mL. Coacervates were reconstituted by mixing PLL and HA solutions to obtain a mass ratio of 1:4 and a final concentration of 10 mg/mL, resulting in net charge neutralization. This solution was diluted 1:10 in 50 mM Tris-Cl pH 8, 7.5\% PEG, and different amounts of MgCl\textsubscript{2} (as indicated). For FRAP experiments, FITC-labeled PLL was spiked into unlabeled PLL at a ratio of 1:500. For FCS experiments and for quantifications of the partition coefficient, Atto647N-labeled PLL was spiked into unlabeled PLL at a ratio of 1:500. Coacervates with typical radii of 2–4 µm were analyzed.

\subsection{DDX4-YFP condensates in living cells} \label{methods:ddx4}
DDX4-YFP was expressed in NIH/3T3 cells using previously described constructs \cite{nott2015phase,muzzopappa2022detecting}. Cells were cultured at 37 °C in a humidified atmosphere with 5\% CO\textsubscript{2} using Dulbecco’s Modified Eagle Medium (Gibco) supplemented with 10\% fetal bovine serum (Life Technologies), 1× Penicillin/Streptomycin (Gibco), 1× Sodium Pyruvate (Gibco) and 1× GlutaMAX Supplement (Gibco). FRAP and FCS experiments were performed in the same cells. Cells were transfected with Lipofectamine 2000 (Invitrogen) according to the manufacturer’s protocol.

\subsection{Calibrated confocal imaging to determine partition coefficients \label{methods:imaging}}
To quantify the partition coefficients for PLL-HA condensates at different MgCl\textsubscript{2} concentrations, coacervates containing Atto647N-labeled PLL were prepared as described above. The samples were subsequently placed on 8-well chambered LabTek coverslips and $z$-stacks of field-of-views containing multiple coacervates were recorded. Confocal imaging was carried out on a Zeiss LSM 880 confocal light scanning microscope (Carl Zeiss, Oberkochen, Germany), equipped with a 63× NA 1.4 oil immersion objective. The same microscope was used to acquire images of solutions of 10 mg/mL PLL where Atto647N-labeled PLL was spiked in at different concentrations ranging from 0 to 200 µM. The same gain and the same HeNe (633 nm) laser settings were used to image each sample. The average intensity values quantified from the concentration series of Atto647N-labeled PLL were used to obtain a calibration curve that was fitted with the following equation
\begin{equation}
    I(c) = a + b(1-e^{-kc}),
\end{equation}
where $I(c)$ is the average intensity and $c$ is the Atto647N-labeled PLL concentration. To quantify partition coefficients, PLL-HA coacervates were segmented in $z$-projected images based on their intensity using the Otsu method \cite{otsu1975threshold}, and the average intensity in the dense and dilute phase was quantified and converted into concentrations using the above-mentioned calibration curve. For each MgCl\textsubscript{2} concentration, multiple coacervates from at least three images were analyzed, and partition coefficients were calculated according to $\Gamma = c_\text{dense}/c_\text{dilute}$, where $c_\text{dense}$ and $c_\text{dilute}$ correspond to the concentrations in the dense and dilute phase, respectively.

\subsection{Fluorescence Correlation Spectroscopy (FCS) \label{methods:fcs}}
For PLL-HA coacervates, FCS experiments were performed on a Zeiss LSM 880 confocal light scanning microscope in Airyscan mode (Carl Zeiss, Oberkochen, Germany), equipped with a 63×NA 1.4 oil immersion objective. The samples were excited by a HeNe 633 nm laser and the emitted fluorescence light passed a dichroic mirror and an LP 645 filter. The Airyscan detector was aligned with a homogeneous solution of 100 µg/mL Atto647N-labeled PLL in 50 mM Tris buffer. Then, 25 µL of coacervate samples containing Atto647N-labeled PLL were placed in 96-well plates. 
FCS experiments were conducted within the first 15 min after pipetting the samples at room temperature to minimize evaporation effects.
Independent measurements were performed in the dilute and dense phase to acquire 30,000,000 time points at a time resolution of 1.23 µs, a bit depth of 16 bits, a gain of 750 and maximum pinhole aperture.

For DDX4-YFP in living NIH/3T3 cells, FCS experiments were performed on a PicoQuant Luminosa confocal light scanning microscope (PicoQuant, Berlin, Germany), equipped with a 60× NA 1.2 water immersion objective and a SPAD. The samples were excited by a 485 nm laser and the emitted fluorescence light passed a dichroic mirror and a 520/35 nm band-pass filter. Experiments were carried out at 37~$^\circ$C.

FCS data were analyzed with a custom-made Python script.
The first part of the measurements was discarded to exclude bleaching of immobile molecules. For PLL-HA, the signals from the 32 detectors of the Airyscan array were averaged. For DDX4-YFP, the signal of the SPAD was used.
Autocorrelation curves were calculated using the multipletau algorithm implemented in the multipletau python package \cite{muller2012tau}. At least ten experiments were analyzed and averaged per condition. Averaged curves for PLL-HA were fitted with the function:
\begin{align}
    G_\text{n}(\tau) = \frac{1}{N} \left(1+\frac{\tau}{\tau_\text{D}}\right)^{-1} \left(1+\left(\frac{w_0}{z_0}\right)^2 \frac{\tau}{\tau_\text{D}}\right)^{-\frac{1}{2}}.
    \label{eq:fcs}
\end{align}
Here, $N$ is the average particle number in the focal volume, $\tau_D$ is the diffusion time, and $w_0$ and $z_0$ are the beam waist in the lateral and axial direction, respectively. The diffusion time is related to the diffusion coefficient $D$ via $\tau_D = w_0^2/4D$. Averaged curves for DDX4-YFP were fitted with the function:
\begin{align}
    G_\text{a}(\tau) = \frac{1}{N} \left(1+\left(\frac{\tau}{\tau_\text{D}}\right)^\alpha\right)^{-1} \left(1+\left(\frac{w_0}{z_0}\right)^2 \left(\frac{\tau}{\tau_\text{D}}\right)^\alpha\right)^{-\frac{1}{2}}.
    \label{eq:fcs2}
\end{align}
Here, $\alpha$ is the anomaly exponent.

The following relationship was used to obtain the ratio of diffusion coefficients $\Delta$
\begin{equation}
    \Delta = \frac{D_\text{out}}{D_\text{in}} = \frac{\tau_\text{D,in}}{\tau_\text{D,out}}.
    \label{eq:fcs-diffusion-time}
\end{equation}
Here, $\tau_\text{D,in}$ and $\tau_\text{D,out}$ denote the diffusion times measured inside and outside of the condensate, respectively. Conveniently, the expression for $\Delta$ is independent of the beam waist.

\subsection{Fluorescence Recovery After Photobleaching (FRAP) \label{methods:frap}}
Fluorescence recovery after photobleaching (FRAP) experiments with PLL-HA coacervates were performed on a Zeiss LSM 710 confocal light scanning microscope (Carl Zeiss, Oberkochen, Germany), equipped with a 63× NA 1.4 oil immersion objective.
Experiments were carried out as described previously \cite{muzzopappa2022detecting}. 
Briefly, 5 µL of each sample were placed on 8-well chambered LabTek coverslips that had been passivated beforehand with 15\% PEG and extensively rinsed.
The experiments were conducted within the first 15 min after pipetting the sample at room temperature to minimize evaporation effects. 
For each PLL-HA experiment, 600 images of $256\times256$ pixels (pixel size 40--75~nm) were acquired every 0.15~s with excitation at 488~nm and a pinhole size of 5.6 Airy units, 60 of them before photobleaching.
A $z$-stack with 2.1~µm steps and a dark image were recorded with the same settings.

FRAP experiments with DDX4-YFP condensates were performed on the same Zeiss LSM 710 microscope with the 63× NA~1.4 oil immersion objective, with excitation at 514~nm, detection centered at 570~nm, a pinhole size of 2.0 Airy units and a photomultiplier detector. Images of $256\times256$ pixels (pixel size 20--44~nm) were acquired every 0.15~s, with 60 frames before photobleaching and 200 to 600 frames in total (300 for most experiments).
Each field of view was also recorded as a $z$-stack with 0.77~$\mu$m steps (8 to 12 planes), and a dark image was recorded with the same settings to estimate the noise level. 

\subsection{Quantitative analysis of confocal FRAP videos \label{methods:video-fitting}}

The recorded image sequences were compared with simulated ones produced based on three components (Fig.~\ref{fig:video-overview}b): the three-dimensional bleach profile, its evolution under the transport model of Section~\ref{sec:model}, and confocal image formation. The quantities $\Gamma$ and $\Delta$ were measured independently and held fixed, while the internal diffusion coefficient $\Din$, the permeability $\kappa$ and the bleaching parameters were adjusted by minimizing a normalized mismatch between simulated and recorded frames.

\subsubsection{Data and fixed inputs}
Each video was corrected for the detector offset, measured on a dark image, and for the lateral drift of the condensates. The bleach regions were read from the acquisition metadata, and pixel noise was modeled as Poisson--Gaussian, calibrated on the prebleach frames and the dark image. For image formation, each condensate was represented as a homogeneous sphere tangent to the coverslip. This sets its axial position for the optical model: the coverslip enters through the imaging depth and the refractive-index-dependent confocal point-spread function, but is not a boundary in the transport model, which uses concentric spherical condensate and dilute-phase domains. The radius $R$ and the depth $z_f$ of the imaging plane were fitted to the prebleach image and, for PLL-HA, to the $z$-stack, and held fixed. The optical parameters were read from the metadata, the sample refractive index was 1.414 for PLL-HA and 1.38 for cells.

\subsubsection{Forward model}
\textit{Confocal point-spread function.} The oil-immersion objective is focused at depth $z_f$ into the sample. The refractive-index mismatch at the coverslip adds a depth-dependent aberration, which expands the focus and shifts it axially. The excitation intensity $h_\mathrm{exc}(\mathbf x)$, at the excitation wavelength, and the detection intensity $h_\mathrm{det}(\mathbf x)$, at the emission wavelength, were computed using vectorial diffraction theory. 
The confocal point-spread function reads
\begin{equation}
    h(\mathbf x)=h_\mathrm{exc}(\mathbf x)\,\bigl[h_\mathrm{det}\ast_{xy}P\bigr](\mathbf x),
    \label{eq:confocal-psf}
\end{equation}
where the pinhole $P$ is a disk with a diameter read from the metadata and $\ast_{xy}$ is a lateral convolution.

\textit{Bleach.} The bleach was described by the light dose received during the scan of the bleach region, which involves the excitation alone. With $S(x,y)$ the indicator of the scanned region, the normalized dose is
\begin{equation}
    B_n(\mathbf x)=\frac{\bigl[S\ast_{xy}h_\mathrm{exc}^{\,n}\bigr](\mathbf x)}
                       {\max_{\mathbf x'}\bigl[S\ast_{xy}h_\mathrm{exc}^{\,n}\bigr](\mathbf x')}.
    \label{eq:bleach-dose}
\end{equation}
where $n\in [0.5,1.5]$ is the effective bleaching order.
The initial concentration of bleached molecules follows a saturating dose response,
\begin{equation}
    c_0(\mathbf x)=(1-f)\,c_\mathrm{eq}(\mathbf x)\bigl[1-e^{-\alpha B_n(\mathbf x)}\bigr],
    \label{eq:bleach-ic}
\end{equation}
with $c_\mathrm{eq}=\Gamma$ in $\Oin$ and 1 in $\Oout$. The three bleaching parameters $\theta_\mathrm{FRAP}=(\alpha,f,n)$, fitted for each condensate, are a dose scale, a residual fluorescent fraction and an effective bleaching order accounting for a nonlinear dose response. As the bleach scans a column through the sample, the dose extends above and below the imaging plane and into the dilute phase.

\textit{Dynamics.} The bleach profile was propagated with the spectral representation of Appendix~\ref{spectral}, which gives the state at all recorded times. 

\textit{Image formation.} A frame records the unbleached molecules, $c_\mathrm{eq}-c(t)$, seen through the point-spread function \eqref{eq:confocal-psf} and integrated over each detector pixel. Applied to the spectral basis, this defines a linear operator $\Psi$ from the state to the image, computed once per condensate, and the simulated frame is
\begin{equation}
    \hat I_t=g\,\Psi\bigl[c_\mathrm{eq}-c(t)\bigr],
    \label{eq:image-model}
\end{equation}
with the gain $g$ fitted on the prebleach image.

\subsubsection{Mismatch and optimization}
The comparison was restricted to a central disk $Z_j$ of radius $0.7\,r_\mathrm{eq}$ which excludes the blurred edge of the condensate and neighboring objects, where $r_\mathrm{eq}$ is the radius measured at the equatorial plane. For condensate $j$, the mismatch is
\begin{equation}
    L_j=\frac{1}{\sigma_j^2\,N_j\,N_t}\sum_t\sum_{\mathbf r\in Z_j}
        \bigl[\hat I_{j,t}(\mathbf r)-h_{j,t}\,u_{j,t}(\mathbf r)\bigr]^2,
    \label{eq:mismatch}
\end{equation}
with $u_{j,t}$ the recorded frame minus background, $N_j$ the number of pixels of $Z_j$, $N_t$ the number of frames, $\sigma_j^2$ the mean pixel variance and $h_{j,t}$ the correction for acquisition-induced bleaching, specified below for each system.

The minimization was nested: for fixed $(\Din,\kappa)$, the bleaching parameters of each condensate were optimized, and an outer L-BFGS minimized the resulting mismatch over $(\log\Din,\log\kappa)$, with an exact gradient obtained by automatic differentiation, from 24 random starting points.

\subsubsection{PLL-HA coacervates}

For each MgCl\textsubscript{2} concentration, $\Gamma$ and $\Delta$ were fixed to the values measured by calibrated imaging and FCS (Fig.~\ref{pll_ha_gamma_eta}; Table~\ref{fit-parameters}, 70~mM value interpolated), with $\Dout=\Delta\Din$. The model considers two diffusive domains, the coacervate with radius $R$ and the surrounding dilute phase with radius $L=32R$. Setting the latter radius to $L=64R$ changed the values of $\Din$ and $\kappa$ by at most 3\%. Acquisition-induced photobleaching was removed via $h_{j,t}=1/f_\mathrm{ab}(t)$, where $f_\mathrm{ab}$ is the mean relative intensity of unbleached condensates present in the video, smoothed by a moving average over 15 frames.

A single $(\Din, \kappa)$ pair was jointly fitted to all half-bleached condensates at each salt concentration (8--11 condensates per condition) by minimizing $\sum_j L_j$. Geometry and bleaching parameters were fitted individually for each condensate. Three reasons motivate this joint fit: First, all condensates at a given salt concentration supposedly comprise the same dense phase, so that $\Din$ and $\kappa$ are expected to be common across condensates, whereas $\kappa^*$ differs with condensate size. Second, in a single condensate, the bleaching and transport parameters may partly compensate each other, whereas condensates of different sizes and bleach profiles constrain $\Din$ and $\kappa$ jointly. Third, the correction of acquisition bleaching, measured on unbleached condensates, leaves in each bleached condensate a residual error of zero mean across condensates, which an individual fit would absorb into condensate-specific transport parameters, whereas the joint fit averages it out.

The joint fit gives a single estimate of $\Din$ and $\kappa$ per salt concentration. To estimate their uncertainty, the fit was repeated $n$ times, each time leaving out one of the $n$ condensates (jackknife). From the spread of these $n$ values, measured on a logarithmic scale by their median absolute deviation, we obtained the standard error.

\subsubsection{DDX4 condensates in living cells}
For each condensate, $\Gamma$ was measured at the end of the video, based on the inner disk of the condensate and the surrounding nucleoplasm. The dilute phase corresponds to the finite nucleoplasm of a living cell, whose bleaching and recovery are not necessarily described by free diffusion in a sphere. Instead of a second diffusive domain, its measured fluorescence signal in a ring around the condensate was fitted by a sum of exponentials, and was set as the outer concentration. $\Delta$, measured by FCS ($\Delta=6$, Fig.~\ref{pll_ha_fcs}b), therefore does not enter the fit of $(\Din,\kappa)$ but was only used to obtain $\rho$ and $M$. A fraction of the fluorescence signal was treated as immobile. Unbleached condensates showed no decrease in fluorescence, so no correction was applied ($h_{j,t}=1$).

DDX4-YFP condensates resided in different cells that might differ from each other. Thus, each of the 19 half-bleached condensates was fitted individually, with its own $(\Din,\kappa)$. The individual fits give a distribution of estimates across condensates, and Table~\ref{fit-parameters} reports its median with a 95\% percentile bootstrap interval.

\bibliography{biblio}

@article{hubatsch2025transport,
  title={Transport kinetics across interfaces between coexisting liquid phases},
  author={Hubatsch, Lars and Bo, Stefano and Harmon, Tyler S and Hyman, Anthony A and Weber, Christoph A and J{\"u}licher, Frank},
  journal={eLife},
  pages={2025},
  year={2025}
}

@article{van2024probing,
  title={Probing the surface charge of condensates using microelectrophoresis},
  author={van Haren, Merlijn HI and Visser, Brent S and Spruijt, Evan},
  journal={Nature Communications},
  volume={15},
  number={1},
  pages={3564},
  year={2024},
  publisher={Nature Publishing Group UK London}
}

@article{majee2024charge,
  title={Charge separation at liquid interfaces},
  author={Majee, Arghya and Weber, Christoph A and J{\"u}licher, Frank},
  journal={Physical Review Research},
  volume={6},
  number={3},
  pages={033138},
  year={2024},
  publisher={APS}
}

@article{zhang2024the,
article_type = {journal},
title = {The exchange dynamics of biomolecular condensates},
author = {Zhang, Yaojun and Pyo, Andrew GT and Kliegman, Ross and Jiang, Yoyo and Brangwynne, Clifford P and Stone, Howard A and Wingreen, Ned S},
editor = {Murugan, Arvind and Cui, Qiang},
volume = 12,
year = 2024,
month = {sep},
pub_date = {2024-09-25},
pages = {RP91680},
citation = {eLife 2024;12:RP91680},
journal = {eLife},
issn = {2050-084X},
publisher = {eLife Sciences Publications, Ltd},
}

@article{scott1951diffusion,
  title={Diffusion through an interface},
  author={Scott, EJ and Tung, LH and Drickamer, HG},
  journal={Journal of Chemical Physics},
  volume={19},
  number={9},
  pages={1075--1078},
  year={1951}
}

@article{novikov2011random,
  title={Random walks with barriers},
  author={Novikov, Dmitry S. and Fieremans, Els and Jensen, Jens H. and Helpern, Joseph A.},
  journal={Nature Physics},
  volume={7},
  number={6},
  pages={508--514},
  year={2011},
  publisher={Nature Publishing Group}
}

@book{brezis2011functional,
  title={Functional Analysis, Sobolev Spaces and Partial Differential Equations},
  author={Br{\'e}zis, Ha{\"\i}m},
  year={2011},
  publisher={Springer}
}

@article{muzzopappa2022detecting,
  title={Detecting and quantifying liquid--liquid phase separation in living cells by model-free calibrated half-bleaching},
  author={Muzzopappa, Fernando and Hummert, Johan and Anfossi, Michela and Tashev, Stanimir Asenov and Herten, Dirk-Peter and Erdel, Fabian},
  journal={Nature Communications},
  volume={13},
  number={1},
  pages={7787},
  year={2022},
  publisher={Nature Publishing Group UK London}
}

@article{bressloff2022probabilistic,
  title={A probabilistic model of diffusion through a semi-permeable barrier},
  author={Bressloff, Paul C},
  journal={Proceedings of the Royal Society A},
  volume={478},
  number={2268},
  pages={20220615},
  year={2022},
  publisher={The Royal Society}
}

@article{mittag2022conceptual,
  title={A conceptual framework for understanding phase separation and addressing open questions and challenges},
  author={Mittag, Tanja and Pappu, Rohit V},
  journal={Molecular cell},
  volume={82},
  number={12},
  pages={2201--2214},
  year={2022},
  publisher={Elsevier}
}

@article{bo2021stochastic,
  title={Stochastic dynamics of single molecules across phase boundaries},
  author={Bo, Stefano and Hubatsch, Lars and Bauermann, Jonathan and Weber, Christoph A and J{\"u}licher, Frank},
  journal={Physical Review Research},
  volume={3},
  number={4},
  pages={043150},
  year={2021},
  publisher={APS}
}

@article{yewdall2021coacervates,
  title={Coacervates as models of membraneless organelles},
  author={Yewdall, N Amy and Andr{\'e}, Alain AM and Lu, Tiemei and Spruijt, Evan},
  journal={Current Opinion in Colloid \& Interface Science},
  volume={52},
  pages={101416},
  year={2021},
  publisher={Elsevier}
}

@article{hubatsch2021quantitative,
  title={Quantitative theory for the diffusive dynamics of liquid condensates},
  author={Hubatsch, Lars and Jawerth, Louise M and Love, Celina and Bauermann, Jonathan and Tang, TY Dora and Bo, Stefano and Hyman, Anthony A and Weber, Christoph A},
  journal={eLife},
  volume={10},
  pages={e68620},
  year={2021},
  publisher={eLife Sciences Publications Limited}
}

@article{erdel2020mouse,
  title={Mouse heterochromatin adopts digital compaction states without showing hallmarks of HP1-driven liquid-liquid phase separation},
  author={Erdel, Fabian and Rademacher, Anne and Vlijm, Rifka and T{\"u}nnermann, Jana and Frank, Lukas and Weinmann, Robin and Schweigert, Elisabeth and Yserentant, Klaus and Hummert, Johan and Bauer, Caroline and others},
  journal={Molecular cell},
  volume={78},
  number={2},
  pages={236--249},
  year={2020},
  publisher={Elsevier}
}

@article{choi2020physical,
  title={Physical principles underlying the complex biology of intracellular phase transitions},
  author={Choi, Jeong-Mo and Holehouse, Alex S and Pappu, Rohit V},
  journal={Annual review of biophysics},
  volume={49},
  number={1},
  pages={107--133},
  year={2020},
  publisher={Annual Reviews}
}

@article{park2020dehydration,
  title={Dehydration entropy drives liquid-liquid phase separation by molecular crowding},
  author={Park, Sohee and Barnes, Ryan and Lin, Yanxian and Jeon, Byoung-jin and Najafi, Saeed and Delaney, Kris T and Fredrickson, Glenn H and Shea, Joan-Emma and Hwang, Dong Soo and Han, Songi},
  journal={Communications Chemistry},
  volume={3},
  number={1},
  pages={83},
  year={2020},
  publisher={Nature Publishing Group UK London}
}

@inproceedings{schafer2020spherical,
  title={Spherical diffusion model with semi-permeable boundary: A transfer function approach},
  author={Sch{\"a}fer, Maximilian and Wicke, Wayan and Haselmayr, Wetner and Rabenstein, Rudolf and Schober, Robert},
  booktitle={ICC 2020-2020 IEEE International Conference on Communications (ICC)},
  pages={1--7},
  year={2020},
  organization={IEEE}
}

@article{Moutal2019,
  title = {Diffusion Across Semi-permeable Barriers: Spectral Properties,  Efficient Computation,  and Applications},
  volume = {81},
  ISSN = {1573-7691},
  url = {http://dx.doi.org/10.1007/s10915-019-01055-5},
  DOI = {10.1007/s10915-019-01055-5},
  number = {3},
  journal = {Journal of Scientific Computing},
  publisher = {Springer Science and Business Media LLC},
  author = {Moutal,  Nicolas and Grebenkov,  Denis},
  year = {2019},
  month = sep,
  pages = {1630–1654}
}

@article{taylor2019quantifying,
  title={Quantifying dynamics in phase-separated condensates using fluorescence recovery after photobleaching},
  author={Taylor, Nicole O and Wei, Ming-Tzo and Stone, Howard A and Brangwynne, Clifford P},
  journal={Biophysical journal},
  volume={117},
  number={7},
  pages={1285--1300},
  year={2019},
  publisher={Elsevier}
}

@article{rapp2018mechanisms,
  title={Mechanisms of diffusion in associative polymer networks: evidence for chain hopping},
  author={Rapp, Peter B and Omar, Ahmad K and Silverman, Bradley R and Wang, Zhen-Gang and Tirrell, David A},
  journal={Journal of the American Chemical Society},
  volume={140},
  number={43},
  pages={14185--14194},
  year={2018},
  publisher={ACS Publications}
}

@article{hamad2018linear,
  title={Linear viscoelasticity and swelling of polyelectrolyte complex coacervates},
  author={Hamad, Fawzi G and Chen, Quan and Colby, Ralph H},
  journal={Macromolecules},
  volume={51},
  number={15},
  pages={5547--5555},
  year={2018},
  publisher={ACS Publications}
}

@article{carr2018modelling,
  title={Modelling mass diffusion for a multi-layer sphere immersed in a semi-infinite medium: application to drug delivery},
  author={Carr, Elliot J and Pontrelli, Giuseppe},
  journal={Mathematical biosciences},
  volume={303},
  pages={1--9},
  year={2018},
  publisher={Elsevier}
}

@article{banani2017biomolecular,
  title={Biomolecular condensates: organizers of cellular biochemistry},
  author={Banani, Salman F and Lee, Hyun O and Hyman, Anthony A and Rosen, Michael K},
  journal={Nature reviews Molecular cell biology},
  volume={18},
  number={5},
  pages={285--298},
  year={2017},
  publisher={Nature Publishing Group UK London}
}

@article{spruijt2013linear,
  title={Linear viscoelasticity of polyelectrolyte complex coacervates},
  author={Spruijt, Evan and Cohen Stuart, Martien A and van der Gucht, Jasper},
  journal={Macromolecules},
  volume={46},
  number={4},
  pages={1633--1641},
  year={2013},
  publisher={ACS Publications}
}

@article{muller2012tau,
  title={Python Multiple-tau Algorithm},
  author={Müller, P},
  journal={Python Package Index},
  year={2012}
}

@article{brangwynne2009germline,
  title={Germline P granules are liquid droplets that localize by controlled dissolution/condensation},
  author={Brangwynne, Clifford P and Eckmann, Christian R and Courson, David S and Rybarska, Agata and Hoege, Carsten and Gharakhani, J{\"o}bin and J{\"u}licher, Frank and Hyman, Anthony A},
  journal={Science},
  volume={324},
  number={5935},
  pages={1729--1732},
  year={2009},
  publisher={American Association for the Advancement of Science}
}

@article{grebenkov2008analytical,
  title={Analytical solution for restricted diffusion in circular and spherical layers under inhomogeneous magnetic fields},
  author={Grebenkov, Denis S},
  journal={The Journal of chemical physics},
  volume={128},
  number={13},
  year={2008},
  publisher={AIP Publishing}
}

@article{rubinstein2001dynamics,
  title={Dynamics of entangled solutions of associating polymers},
  author={Rubinstein, Michael and Semenov, Alexander N},
  journal={Macromolecules},
  volume={34},
  number={4},
  pages={1058--1068},
  year={2001},
  publisher={ACS Publications}
}

@article{powles1992exact,
  title={Exact analytic solutions for diffusion impeded by an infinite array of partially permeable barriers},
  author={Powles, Jack G and Mallett, Margaret JD and Rickayzen, Gerald and Evans, WAB},
  journal={Proceedings of the Royal Society of London. Series A: Mathematical and Physical Sciences},
  volume={436},
  number={1897},
  pages={391--403},
  year={1992},
  publisher={The Royal Society London}
}

@article{otsu1975threshold,
  title={A threshold selection method from gray-level histograms},
  author={Otsu, Nobuyuki and others},
  journal={Automatica},
  volume={11},
  number={285-296},
  pages={23--27},
  year={1975}
}

@article{nott2015phase,
  author  = {Nott, Timothy J. and Petsalaki, Evangelia and Farber, Patrick and Jervis, Dylan and Fussner, Eden and Plochowietz, Anne and Craggs, Timothy D. and Bazett-Jones, David P. and Pawson, Tony and Forman-Kay, Julie D. and Baldwin, Andrew J.},
  title   = {Phase Transition of a Disordered Nuage Protein Generates Environmentally Responsive Membraneless Organelles},
  journal = {Molecular Cell},
  year    = {2015},
  volume  = {57},
  number  = {5},
  pages   = {936--947},
  doi     = {10.1016/j.molcel.2015.01.013},
  pmid    = {25747659}
}

\end{document}